\documentclass[11pt]{article}
\usepackage{geometry}
\usepackage{amsmath, amssymb, amsthm}
\usepackage{graphicx}
\usepackage{bm}
\usepackage[linesnumbered,ruled]{algorithm2e}
\usepackage{fullpage}
\usepackage{mathtools}
\usepackage{url}
\newtheorem{theorem}{Theorem}
\newtheorem{lemma}{Lemma}

\newtheorem{claim}{Claim}
\newtheorem{definition}{Definition}
\SetKwRepeat{Do}{do}{while}
\title{Non-uniformly Stable Common Independent Sets}
\author{Naoyuki Kamiyama\thanks{%
This work was supported by JST ERATO Grant Number JPMJER2301, Japan.}}
\date{Institute of Mathematics for Industry \\ Kyushu University \\ Fukuoka, Japan \\
\url{kamiyama@imi.kyushu-u.ac.jp}}

\begin{document}

\maketitle

\begin{abstract}
In this paper, we consider 
a matroid generalization of the stable matching 
problem. 
In particular, we consider the setting where 
preferences may contain ties. 
For this generalization, we propose a polynomial-time 
algorithm 
for the problem of checking the existence 
of a common independent set  
satisfying non-uniform stability, which is 
a common generalization of super-stability and 
strong stability. 
\end{abstract} 

\section{Introduction}

The stable matching problem, which was proposed by 
Gale and Shapley~\cite{GaleS62}, is one of the most 
central problems in the study of matching under preference~\cite{Manlove13}.  
In this problem, we are given two disjoint groups of agents. 
Each agent has a preference over a subset of agents in the other group. 
A matching between two groups is said to be stable if there does not exist 
an unmatched pair of agents who 
have incentives to deviate from the current matching. 
In other words, there does not exist 
an unmatched pair of agents who 
prefer the new partner to the current partner. 
Gale and Shapley~\cite{GaleS62} proved that 
when the preferences are strict (i.e., they do not contain ties), 
a stable matching always exists and 
we can find a stable matching in polynomial time. 

In this paper, we consider the following 
two directions of generalization of the stable matching problem. 
The first direction is introducing ties in preferences. 
In other words, agents may be indifferent between potential partners. 
The second direction is generalization to matroids.
In the stable matching problem, an instance is often represented by 
a bipartite graph. 
In this direction, a matching in a bipartite graph 
is generalized to a 
common independent set of two matroids. 

In the stable matching problem with ties, 
Irving~\cite{Irving94} proposed the following three properties of a matching 
(see, e.g., \cite{IwamaM08} and 
\cite[Chapter~3]{Manlove13} for a survey of the stable matching problem 
with ties).  
The first property is weak stability. 
This property guarantees that 
there does not exist an unmatched pair of agents such that 
both prefer the other agent in the pair to the 
current partner. 
For this property, 
Irving~\cite{Irving94} proved 
that a weakly stable matching always exists, and 
a weakly stable matching can be found
in polynomial time by using the algorithm 
of Gale and Shapley~\cite{GaleS62} with tie-breaking.
The second property is super-stability.
This property guarantees that 
there does not exist an unmatched pair of agents such that
both weakly prefer the other agent in the pair to the current partner
(i.e., the other agent in the pair 
is not worse than the current partner for any agent in the pair). 
The third property is strong stability.
This stability guarantees that 
there does not exist an unmatched pair of agents such that
both weakly prefer the other agent in the pair to the current partner, and
at least one of the agents in the pair 
prefers the other agent in the pair 
to the current partner. 
It is known that a super-stable matching and a strongly stable matching 
may not exist~\cite{Irving94}. 
For the problem of checking the existence of 
a super-stable matching and 
a strongly stable matching in the one-to-one setting, 
Irving~\cite{Irving94} 
proposed 
polynomial-time algorithms 
(see also \cite{Manlove99}). 
In the many-to-one setting, 
Irving, Manlove, and Scott~\cite{IrvingMS00} proposed 
a polynomial-time algorithm for 
the super-stable matching problem, and 
the papers~\cite{IrvingMS03,KavithaMMP07}
gave polynomial-time 
algorithms for the 
strongly stable matching
problem. 
Scott~\cite{Scott05} considered 
the super-stable matching problem
in the many-to-many setting.
In addition, 
in the many-to-many setting,  
the papers~\cite{ChenG10,Kunysz19,Malhotra04} considered 
the strongly stable matching problem.
Olaosebikan and Manlove~\cite{OlaosebikanM20}
considered super-stability 
in the student-project allocation problem,
which is a variant of the stable matching problem with hierarchical 
capacity constraints.
Olaosebikan and Manlove~\cite{OlaosebikanM22}
considered strong stability
in the student-project allocation problem. 

For the stable matching problem with ties, 
Kamiyama~\cite{Kamiyama25+}
introduced non-uniform stability, which is 
a common generalization of 
super-stability and strong stability. 
More concretely, the pairs of agents 
is divided into two groups. 
Then in the definition of non-uniform stability, 
the pairs in one group have to 
satisfy the condition for super-stability, and 
those in the other group have to 
satisfy the condition for strong stability.  
Kamiyama~\cite{Kamiyama25+} proved that 
the existence of a non-uniformly stable matching 
can be checked in polynomial time. 

Fleiner~\cite{Fleiner03} introduced a matroid generalization of 
the stable matching problem with strict preferences.
In addition, Fleiner~\cite{Fleiner03}
proved that 
there always exists a stable solution in this generalization and 
we can find a stable solution in polynomial time if we are given 
polynomial-time independence oracles for the matroids. 
For a matroid generalization of 
the super-stable matching problem, 
Kamiyama~\cite{Kamiyama22} proposed 
a polynomial-time algorithm 
for checking 
the existence of a super-stable solution.
Furthermore, 
Kamiyama~\cite{Kamiyama25} proposed 
a polynomial-time algorithm 
for 
the many-to-one strongly stable matching 
problem under a matroid constraint. 

In this paper, we consider 
a non-uniformly stable common independent set of 
two matroids, i.e., 
we consider both of the above two directions. 
We propose a polynomial-time 
algorithm 
for the problem of checking the existence of 
a non-uniformly stable 
common independent set of two matroids.
The result of this paper is a common generalization of 
\cite{Kamiyama22} (a matroid generalization of 
the super-stable matching problem),
\cite{Kamiyama25} (the many-to-one strongly stable matching 
problem under a matroid constraint), and 
\cite{Kamiyama25+} (non-uniformly stable matchings 
in the one-to-one setting). 
Thus, our proof is based on the techniques in these papers. 
However, 
the extension of the many-to-one case to the many-to-many case
is generally a non-trivial task (see, e.g., \cite{ChenG10}). 
Actually, in some parts of our proof (e.g., the proof of Lemma~\ref{lemma:key}),
we need careful discussion.

\section{Preliminaries} 

For each non-negative integer $z$, 
we define $[z]$ as the set of integers $i$ 
such that $1 \le i \le z$. 
For each finite set $X$ and each element $x$, 
we define $X + x \coloneqq X \cup \{x\}$ and 
$X - x \coloneqq X \setminus \{x\}$. 

In a finite simple
directed graph, $uv$ denotes  
an arc from a vertex $u$ to a vertex $v$. 
A directed cycle $C$ in a finite simple directed graph is said to be 
\emph{simple} if $C$ passes through each vertex at most once. 
For each simple directed cycle $C$ 
in a finite simple directed graph, an arc $a$ that $C$ does not pass through 
is called a \emph{shortcut arc for $C$} if 
$C + a$ contains a simple directed cycle that is different 
from $C$, where we do not distinguish between 
the simple directed cycle $C$ and the set of arcs 
that $C$ passes through. 
It is not difficult to see that when  
there exists a directed 
cycle in a finite simple directed graph,  
there exists a simple directed cycle for which 
there does not exist a shortcut arc. 

An ordered pair ${\cal M} = (U,{\cal I})$ of a finite set $U$ and 
a non-empty family ${\cal I}$ 
of subsets of $U$ is called a \emph{matroid}
if the following conditions are satisfied for every pair of 
subsets $I,J \subseteq U$. 
\begin{description}
\item[(I1)]
If $I \subseteq J$ and $J \in {\cal I}$, then 
$I \in {\cal I}$. 
\item[(I2)]
If $I,J \in {\cal I}$ and 
$|I| < |J|$, then there exists an element $u \in J \setminus I$ such that 
$I + u \in {\cal I}$. 
\end{description} 
An element in ${\cal I}$ is called an 
\emph{independent set of ${\cal M}$}. 

When we consider an algorithm for a matroid ${\cal M} = (U,{\cal I})$, 
we assume that 
we can determine whether $I \in {\cal I}$ 
for each subset $I \subseteq U$ in time bounded by a polynomial in 
$|U|$.

Let ${\cal M}_1 = (U, {\cal I}_1)$ and ${\cal M}_2 = (U, {\cal I}_2)$ 
be matroids. 
Then a subset $I \subseteq U$ is called 
a \emph{common independent set of ${\cal M}_1, {\cal M}_2$} 
if $I \in {\cal I}_1 \cap {\cal I}_2$. 

In this paper, we are given two matroids 
${\cal M}_D = (E, {\cal I}_D)$ and 
${\cal M}_H = (E, {\cal I}_H)$, and 
transitive and complete binary relations 
$\succsim_D$ and $\succsim_H$ on $E$.
(The term ``complete'' means that, 
for every element $S \in \{D,H\}$ and 
every pair of elements $e,f \in E$, 
at least one of $e \succsim_S f$, $f \succsim_S e$
holds.) 
For each element $S \in \{D,H\}$ and 
each pair of elements $e,f \in E$, 
if $e \succsim_S f$ and $f \not\succsim_S e$ 
(resp.\ $e \succsim_S f$ and $f \succsim_S e$), 
then 
we write $e \succ_S f$ (resp.\ $e \sim_S f$). 
Furthermore, we are given disjoint subsets $E_1,E_2 \subseteq E$
such that $E_1 \cup E_2 = E$. 
We assume that, for 
every element $e \in E$
and every element $S \in \{D,H\}$, 
we have $\{e\} \in {\cal I}_S$. 

\begin{definition}
Let $I$ be a common independent set of 
${\cal M}_D, {\cal M}_H$, and let $e$ be an element in $E \setminus I$. 
Let $S$ be an element in $\{D,H\}$.
Then we say that 
$e$ \emph{weakly blocks} $I$ on ${\cal M}_S$ if 
one of the following conditions is satisfied. 
\begin{description}
\item[(W1)]
$I + e \in {\cal I}_S$.
\item[(W2)]
$I + e \notin {\cal I}_S$, and 
there exists an element $f \in I$ such that 
$e \succsim_S f$ and $I + e - f \in {\cal I}_S$.
\end{description}
We say that 
$e$ \emph{strongly blocks} $I$ on ${\cal M}_S$ if 
one of the following conditions is satisfied. 
\begin{description}
\item[(S1)]
$I + e \in {\cal I}_S$.
\item[(S2)]
$I + e \notin {\cal I}_S$, and 
there exists an element $f \in I$ such that 
$e \succ_S f$ and $I + e - f \in {\cal I}_S$.
\end{description}
\end{definition} 

\begin{definition}
Let $I$ be a common independent set of ${\cal M}_D, {\cal M}_H$, 
and let $e$ be an element in $E \setminus I$. 
Then we say that $e$ \emph{blocks} $I$ if the following 
conditions are satisfied. 
\begin{itemize}
\item
If $e \in E_1$, then
$e$ weakly blocks $I$ on $\mathcal{M}_S$ for every element $S \in \{D,H\}$. 
\item
If $e \in E_2$, then 
$e$ weakly blocks $I$ on $\mathcal{M}_S$ for every element $S \in \{D,H\}$, and 
there exists an element $S \in \{D,H\}$ such that 
$e$ strongly blocks $I$ on ${\cal M}_S$. 
\end{itemize}
\end{definition} 

\begin{definition}
A common independent set $I$ of ${\cal M}_D, {\cal M}_H$
is said to be \emph{non-uniformly stable} if no 
element in $E \setminus I$ blocks $I$. 
\end{definition}

\subsection{Basics of matroids} 

Throughout this subsection, let ${\cal M} = (U,{\cal I})$ be a matroid. 

An inclusion-wise maximal independent set of ${\cal M}$ 
is called a \emph{base of ${\cal M}$}. 
Then the condition (I2) implies that 
all the bases of ${\cal M}$ have the same size. 
In addition, (I2) implies that, 
for every independent set $I$ of ${\cal M}$, 
there exists a base $B$ of ${\cal M}$ such that 
$I \subseteq B$. 

A subset $C \subseteq U$ 
such that $C \notin {\cal I}$ 
is called a \emph{dependent set of ${\cal M}$}. 
An inclusion-wise minimal dependent set of ${\cal M}$ 
is called a \emph{circuit of ${\cal M}$}. 

\begin{lemma}[{See, e.g., \cite[Page~15, Exercise~14]{Oxley11}}] \label{lemma:elimination}
Let $C_1,C_2$ be distinct circuits of ${\cal M}$ such that 
$C_1 \cap C_2 \neq \emptyset$. 
Then for every of element $e \in C_1 \cap C_2$ and 
every element $f \in C_1 \setminus C_2$, 
there exists a circuit $C$ of ${\cal M}$ 
such that $f \in C \subseteq (C_1 \cup C_2) - e$.
\end{lemma}

Let $X$ be a subset of $U$.
Then we define ${\cal I}|X$ as the family of 
subsets $I\subseteq X$ such that $I \in {\cal I}$. 
Furthermore, we define ${\cal M}|X := (X, {\cal I}|X)$.
Then it is known that 
${\cal M}|X$ is a matroid
(see, e.g., \cite[p.20]{Oxley11}). 
Define ${\bf rk}_{{\cal M}}(X)$ as the size of a base of ${\cal M}|X$. 

For each subset $X \subseteq U$, 
we define 
${\bf cl}_{{\cal M}}(X)$ as the set of elements 
$u \in U$ such that ${\bf rk}_{{\cal M}}(X + u) = {\bf rk}_{{\cal M}}(X)$. 
Clearly, $X \subseteq {\bf cl}_{{\cal M}}(X)$ for every subset $X \subseteq U$. 

\begin{lemma}[{See, e.g., \cite[Proposition~1.4.11]{Oxley11}}] \label{lemma:closure}
Let $X$ be a subset of $U$. 
Then for every element $u \in U \setminus X$, 
$u \in {\bf cl}_{{\cal M}}(X)$ if and only if 
there exists a circuit $C$ of ${\cal M}$ 
such that 
$u \in C \subseteq X + u$.
\end{lemma}

Lemma~\ref{lemma:closure} implies that
we have
${\bf cl}_{\cal M}(X) \subseteq {\bf cl}_{\cal M}(Y)$
for every pair 
of subsets $X, Y \subseteq U$
such that $X \subseteq Y$.
Thus, for every pair 
of subsets $X, Y \subseteq U$, 
if $X \subseteq Y$ and $Y \setminus X \subseteq {\bf cl}_{\cal M}(X)$, 
then 
we have ${\bf rk}_{\cal M}(X) = {\bf rk}_{\cal M}(Y)$.
In particular, when $Y = {\bf cl}_{\cal M}(X)$, we have 
${\bf rk}_{{\cal M}}(X) = {\bf rk}_{{\cal M}}({\bf cl}_{{\cal M}}(X))$
(see, e.g., \cite[Lemma~1.4.2]{Oxley11}).

Let $X$ be a subset of $U$.
Then 
we define 
${\cal I}/X$ as the family of 
subsets $I \subseteq U \setminus X$ 
such that 
$|I| = {\bf rk}_{{\cal M}}(I \cup X) - {\bf rk}_{{\cal M}}(X)$.
Define 
${\cal M}/X := (U \setminus X, {\cal I}/X)$. 
Then it is known that ${\cal M}/X$ is a matroid
(see, e.g., \cite[Proposition~3.1.6]{Oxley11}).

\begin{lemma}[{See, e.g., \cite[Proposition~3.1.25]{Oxley11}}] \label{lemma:minor}
For every pair of disjoint subsets $X,Y \subseteq U$, 
\begin{equation*}
({\cal M}/X) / Y = {\cal M}/ (X \cup Y), \ \ \ 
({\cal M}/X) | Y = ({\cal M}| (X \cup Y)) / X.
\end{equation*}
\end{lemma}

\begin{lemma}[{See, e.g., \cite[Proposition~3.1.7]{Oxley11}}] \label{lemma:contraction}
Let $X$ be a subset of $U$. Let $B$ be a base of ${\cal M}|X$.
Then for every subset $I \subseteq U \setminus X$, 
$I$ is an independent set 
{\rm (}resp.\ a base{\rm )}
of ${\cal M}/ X$
if and only if 
$I \cup B$ is an independent set 
{\rm (}resp.\ a base{\rm )}
of ${\cal M}$. 
\end{lemma} 

Lemma~\ref{lemma:closure} implies that, 
for every independent set $I$ of ${\cal M}$ and 
every element $u \in {\bf cl}_{{\cal M}}(I) \setminus I$, 
there exists a circuit $C$ of ${\cal M}$ such that 
$u \in C \subseteq I + u$.
In addition, 
Lemma~\ref{lemma:elimination}
implies that 
such a circuit $C$ is uniquely determined. 
We call this circuit the {\it fundamental circuit of $I$
and $u$ in ${\cal M}$}, and this circuit is 
denoted by 
$\mathbb{C}_{{\cal M}}(u,I)$.
It is known that, 
for every independent set $I$ of 
${\cal M}$ and 
every element $u \in {\bf cl}_{{\cal M}}(I) \setminus I$, 
$\mathbb{C}_{{\cal M}}(u,I)$ coincides with the set of 
elements $v \in I  + u$ such that  
$I + u - v \in {\cal I}$
(see, e.g., \cite[Page~20, Exercise~5]{Oxley11}). 
For each independent set $I$ of ${\cal M}$ and 
each element $u \in {\bf cl}_{{\cal M}}(I) \setminus I$, 
we define $\mathbb{D}_{{\cal M}}(u,I) := \mathbb{C}_{{\cal M}}(u,I) - u$. 
We can rewrite (W2) (resp.\ (S2)) as follows.  
\begin{center}
$e \in {\bf cl}_{{\cal M}_S}(I)$, and 
there exists an element $f \in \mathbb{D}_{{\cal M}_S}(e,I)$ such that 
$e \succsim_S f$
(resp.\ 
$e \succ_S f$).
\end{center}
Furthermore, Lemma~\ref{lemma:closure}
implies that (W1) and (S1) 
can be rewritten as $e \notin {\bf cl}_{{\cal M}_S}(I)$. 

\begin{lemma}[{Iri and Tomizawa~\cite[Lemma~2]{IriT76}}] \label{lemma:IriT76_sequence}
Let $I$ be an independent set of ${\cal M}$.
Furthermore, we assume that we are given 
distinct elements $u_1,u_2,\ldots,u_{\ell} \in {\bf cl}_{\cal M}(I) \setminus I$ 
and $v_1,v_2,\ldots,v_{\ell} \in I$ satisfying the 
following 
conditions.
\begin{itemize}
\item
$v_i \in \mathbb{C}_{\cal M}(u_i,I)$
for every integer $i \in [\ell]$. 
\item
$v_i \notin \mathbb{C}_{\cal M}(u_j,I)$
holds for every pair of integers $i,j \in [\ell]$ such that $i < j$.
\end{itemize}
Define $J := (I \cup \{u_1,u_2,\ldots,u_{\ell}\}) \setminus \{v_1,v_2,\ldots,v_{\ell}\}$.
Then $J \in {\cal I}$ and 
${\bf cl}_{\cal M}(I) = {\bf cl}_{\cal M}(J)$. 
\end{lemma}

Lemma~\ref{lemma:IriT76_sequence} implies that
$x \in {\bf cl}_{\cal M}(I + u - v)$
for every independent set $I$ of ${\cal M}$, 
every element $u \in {\bf cl}_{\cal M}(I) \setminus I$,
every element $v \in \mathbb{C}_{\cal M}(u,I)$, and 
every element $x \in ({\bf cl}_{\cal M}(I) \setminus I) - u$. 

\begin{lemma}[{See, e.g., \cite[Lemma~3]{IriT76}}] \label{lemma:IriT76_auxiliary}
Let $I$ be an independent set of ${\cal M}$.
In addition, let $u$ and $v$ 
be an element 
in ${\bf cl}_{\cal M}(I) \setminus I$ and 
an element in $\mathbb{C}_{\cal M}(u,I)$, respectively. 
Define $J := I + u - v$. 
Let $x$ be 
an element in $({\bf cl}_{\cal M}(I) \setminus I) - u$, 
and let $y$ be 
an element in $I - v$.
Then the following statements hold.
\begin{description}
\item[(i)]
If $v \in \mathbb{C}_{\cal M}(x, I)$, 
then $u \in \mathbb{C}_{\cal M}(x, J)$. 
\item[(ii)]
If $v \notin \mathbb{C}_{\cal M}(x, I)$, 
then $u \notin \mathbb{C}_{\cal M}(x, J)$. 
\item[(iii)]
Assume that $y \in \mathbb{C}_{\cal M}(x,I)$, and 
we have at least one of 
$y \notin \mathbb{C}_{\cal M}(u,I)$, 
$v \notin \mathbb{C}_{\cal M}(x,I)$. 
Then
$y \in \mathbb{C}_{\cal M}(x,J)$.
\item[(iv)]
Assume that 
$y \notin \mathbb{C}_{\cal M}(x,I)$,
and 
we have at least one of 
$y \notin \mathbb{C}_{\cal M}(u,I)$, 
$v \notin \mathbb{C}_{\cal M}(x,I)$.
Then 
$y \notin \mathbb{C}_{\cal M}(x,J)$.
\end{description}
\end{lemma} 

The following lemma plays an important role in this paper. 

\begin{lemma}[{See, e.g., \cite[Lemma~2.4]{Kamiyama22}}] \label{lemma:circuit_union}
Let $C, C_1,C_2,\ldots,C_{\ell}$ be circuits of ${\cal M}$.
In addition, we assume that 
distinct elements $u_1,u_2,\ldots,u_{\ell}, v \in U$ satisfy the following 
conditions. 
\begin{description}
\item[(U1)]
$u_i \in C \cap C_i$  
for every integer $i \in [\ell]$. 
\item[(U2)]
$u_{i} \notin C_{j}$ holds for every 
pair of distinct integers $i,j \in [\ell]$. 
\item[(U3)]
$v \in C \setminus (C_1 \cup C_2 \cup \cdots \cup C_{\ell})$.  
\end{description}
Then there exists a circuit $C^{\prime}$ of ${\cal M}$ 
such that 
\begin{equation*}
C^{\prime} \subseteq \big(C \cup C_1 \cup C_2 \cup \cdots \cup C_{\ell}\big) 
\setminus \{u_1,u_2,\ldots,u_{\ell}\}.
\end{equation*} 
\end{lemma}

\subsection{Direct sum of matroids} 

Assume that we are given matroids 
${\cal M}_1 = (U_1, {\cal I}_1), {\cal M}_2 = (U_2, {\cal I}_2),\dots,
{\cal M}_{\ell} = (U_{\ell}, {\cal I}_{\ell})$
such that 
$U_1,U_2,\dots,U_{\ell}$ are pairwise disjoint. 
Define the \emph{direct sum} ${\cal M} = (U, {\cal I})$ of 
${\cal M}_1, {\cal M}_2, \dots, {\cal M}_{\ell}$
as follows. 
Define $U$ as $U_1 \cup U_2 \cup \dots \cup U_{\ell}$.
In addition, we define ${\cal I}$ as the family of subsets 
$I \subseteq U$ such that
$I \cap U_i \in {\cal I}_i$
for every integer $i \in [\ell]$. 
Then it is known that 
${\cal M}$ is a matroid
(see, e.g., \cite[Proposition~4.2.12]{Oxley11}).

\subsection{Common independent sets}

Throughout this subsection, 
let ${\cal M}_1 = (U, {\cal I}_1)$ and ${\cal M}_2 = (U, {\cal I}_2)$ 
be matroids. 
It is known that 
then we can find a maximum-size 
common independent set of 
${\cal M}_1, {\cal M}_2$ in time bounded by a 
polynomial in $|U|$ (see, e.g., \cite{Blikstad21,Cunningham86}). 

Define the function $\mu_{{\cal M}_1{\cal M}_2}$ on $2^U$ by  
\begin{equation*}
\mu_{{\cal M}_1{\cal M}_2}(X) \coloneqq 
{\bf rk}_{{\cal M}_1}(U \setminus X) + 
{\bf rk}_{{\cal M}_2}(X). 
\end{equation*}

\begin{lemma}[{Edmonds~\cite{Edmonds70}}] \label{lemma:Edmonds70}
\begin{equation} \label{eq:Edmonds70}
\max\{|I| \mid
I \in {\cal I}_1 \cap {\cal I}_2\}
= 
\min\{\mu_{{\cal M}_1{\cal M}_2}(X)
\mid X \subseteq U\}. 
\end{equation}
\end{lemma}
Since $\mu_{{\cal M}_1{\cal M}_2}$ is submodular
(see, e.g., \cite[Lemma 1.3.1]{Oxley11}), 
it is known that 
an inclusion-wise minimal minimizer 
of the right-hand side of \eqref{eq:Edmonds70} 
is uniquely determined.
Furthermore, we can find 
the inclusion-wise minimal minimizer 
of the right-hand side of \eqref{eq:Edmonds70} 
in time bounded by a polynomial in $|U|$ 
(see, e.g., \cite[Note~10.12]{Murota03}). 
We call the inclusion-wise minimal minimizer 
of the right-hand side of \eqref{eq:Edmonds70} 
the \emph{critical subset of ${\cal M}_1,{\cal M}_2$}. 

\begin{definition}
For each common independent set $I$ of ${\cal M}_1,{\cal M}_2$, 
we define the simple directed graph 
${\bf G}_{{\cal M}_1{\cal M}_2}(I) = 
(U, {\bf A}_{{\cal M}_1{\cal M}_2}(I))$ 
as follows. 
\begin{itemize}
\item
Let $u$ be an element in $U\setminus I$, and 
let $v$ be an element in $I$. 
\begin{itemize}
\item 
$vu \in {\bf A}_{{\cal M}_1{\cal M}_2}(I)$ 
if and only if $u \in {\bf cl}_{{\cal M}_1}(I)$ and 
$v \in \mathbb{C}_{{\cal M}_1}(u,I)$. 
\item
$uv \in {\bf A}_{{\cal M}_1{\cal M}_2}(I)$ 
if and only if $u \in {\bf cl}_{{\cal M}_2}(I)$ and 
$v \in \mathbb{C}_{{\cal M}_2}(u,I)$. 
\end{itemize}
\end{itemize}
\end{definition}

Let $I$ be  
a common independent set of ${\cal M}_1, {\cal M}_2$.
Furthermore, let $C$ be a simple directed cycle in 
${\bf G}_{{\cal M}_1{\cal M}_2}(I)$. 
Let $X$ (resp.\ $Y$) be 
the set of elements in $U \setminus I$ (resp.\ 
$I$) that $C$ passes through. 
Then we define 
$I \ominus C := (I \cup X) \setminus Y$. 

\begin{lemma} \label{lemma:IriT76_intersection}
Let $I$ be a common independent set of ${\cal M}_1, {\cal M}_2$. 
Let $C$ be a simple directed cycle in 
${\bf G}_{{\cal M}_1{\cal M}_2}(I)$ such that 
there does not exist a shortcut arc in 
${\bf A}_{{\cal M}_1{\cal M}_2}(I)$ for $C$. 
Then $I \ominus C$ is a common independent set of ${\cal M}_1, {\cal M}_2$. 
\end{lemma} 
\begin{proof} 
This lemma immediately follows from Lemma~\ref{lemma:IriT76_sequence}. 
\end{proof} 

\section{Algorithm} 

For notional simplicity, 
we define 
\begin{equation*}
\mathbb{C}_S(\cdot, \cdot) \coloneqq \mathbb{C}_{{\cal M}_S}(\cdot, \cdot), \ \ 
{\bf rk}_S(\cdot) \coloneqq {\bf rk}_{{\cal M}_S}(\cdot), \ \ 
{\bf cl}_S(\cdot) \coloneqq {\bf cl}_{{\cal M}_S}(\cdot)
\end{equation*}
for each element $S \in \{D,H\}$. 
For each element $S \in \{D, H\}$ and each 
subset $F \subseteq E$, 
we define 
\begin{equation*}
\begin{split}
{\bf he}_S(F) & \coloneqq 
\{e \in F \mid \mbox{$e \succsim_S f$ 
for every element $f \in F$}\} \\
{\bf ta}_S(F) & \coloneqq 
\{e \in F \mid \mbox{$f \succsim_S e$ 
for every element $f \in F$}\}.
\end{split}
\end{equation*}

\begin{definition}
For each common independent set $I$ of ${\cal M}_D,{\cal M}_H$, 
we define ${\bf block}(I)$ as the set of 
elements $e \in E \setminus I$ 
such that $e \in {\bf cl}_H(I)$ and 
$e$ blocks $I$.
\end{definition} 

\subsection{Subroutines} 

For each subset $F \subseteq E$, we define
the matroids ${\cal M}_D\langle F \rangle, {\cal M}_H\langle F \rangle$ and 
the subset ${\rm Ch}_D(F) \subseteq F$ as 
the outputs of 
Algorithms~\ref{alg:matroid_D}, 
\ref{alg:matroid_H}, and \ref{alg:choice_D}, 
respectively. 

Although facts similar to the lemmas in this subsection
were proved in \cite{Kamiyama22,Kamiyama25}, we give their proofs 
for completeness. 

\begin{algorithm}[ht]
Define $N_0 \coloneqq F$ and ${\cal M}_0^{\prime} \coloneqq {\cal M}_D$.
Set $s \coloneqq 0$.\\
If $N_0 = \emptyset$, then output the matroid $(\emptyset,\{\emptyset\})$ as 
${\cal M}_D\langle F \rangle$ and halt.\\
\While{$N_s\neq \emptyset$}
{
   Set $s \coloneqq s+1$.\\
   Define $T_s \coloneqq {\bf he}_D(N_{s-1})$, $N_s \coloneqq N_{s-1} \setminus T_s$,
   ${\cal M}_s \coloneqq {\cal M}_{s-1}^{\prime}|T_s$, and 
   ${\cal M}_s^{\prime} \coloneqq {\cal M}_{s-1}^{\prime} / T_s$.  
}
Output the direct sum of 
${\cal M}_1,{\cal M}_2, \dots, {\cal M}_s$ as ${\cal M}_D\langle F \rangle$
and halt.
\caption{Algorithm for defining ${\cal M}_D\langle F \rangle$}
\label{alg:matroid_D}
\end{algorithm}

\begin{lemma} \label{lemma:base_matroid_D}
Let $F$ be a subset of $E$.
Let $B$ be a base of ${\cal M}_D\langle F \rangle$. 
Assume that Algorithm~\ref{alg:matroid_D} halts when $s = k$, and 
we use the notation in Algorithm~\ref{alg:matroid_D}. 
Define $F_s \coloneqq \bigcup_{j = 1}^sT_j$ 
and $B_s \coloneqq B \cap F_s$ 
for each integer $s \in [k]$. 
Then $B_s$ is a base of ${\cal M}_D|F_s$
for every integer $s \in [k]$.
\end{lemma}
\begin{proof}
Notice that since $B$ is a base of ${\cal M}_D\langle F \rangle$, 
$B \cap T_s$ is a base of ${\cal M}_s$ for every 
integer $s \in [k]$. 
Thus, since ${\cal M}_1 = {\cal M}_D|T_1$ and  
$B \cap T_1$ is a base of ${\cal M}_1$, 
$B_1$ is a base of ${\cal M}_D|F_1$. 

Let $s$ be an integer in $[k-1]$.
Assume that 
$B_s$ is a base of ${\cal M}_D|F_s$. 
Notice that ${\cal M}_D|F_s=({\cal M}_D|F_{s+1})|F_s$ and 
Lemma~\ref{lemma:minor} implies that 
\begin{equation*}
{\cal M}_{s+1} = {\cal M}_s^{\prime}|T_{s+1} = 
({\cal M}_D/F_s)|T_{s+1} = ({\cal M}_D|F_{s+1})/F_s.
\end{equation*} 
Since 
$B \cap T_{s+1}$ is a base of ${\cal M}_{s+1}$, 
Lemma~\ref{lemma:contraction} implies that 
$B_{s+1} = B_s \cup (B \cap T_{s+1})$ 
is a base of ${\cal M}_D|F_{s+1}$. 
This completes the proof. 
\end{proof} 

Lemma~\ref{lemma:base_matroid_D}
implies that, for every subset $F \subseteq E$, 
every base $B$ of ${\cal M}_D\langle F \rangle$, and 
every element $e \in F \setminus B$, 
we have $e \in {\bf cl}_D(B)$.

\begin{lemma} \label{lemma:circuit_matroid_D}
Let $F$ be a subset of $E$.
Let $B$ be a base of ${\cal M}_D\langle F \rangle$.
Let $e$ be an element in $F \setminus B$. 
\begin{description}
\item[(i)]
$e \sim_D f$ for every element 
$f \in \mathbb{C}_{{\cal M}_D\langle F \rangle}(e,B)$.
\item[(ii)]
$\mathbb{C}_{{\cal M}_D\langle F \rangle}(e,B) \subseteq 
\mathbb{C}_D(e,B)$. 
\item[(iii)]
For every element 
$f \in \mathbb{C}_D(e,B) \setminus \mathbb{C}_{{\cal M}_D\langle F \rangle}(e,B)$, 
we have $f \succ_D e$.
\end{description}
\end{lemma} 
\begin{proof}
In this proof, we use the notation in Algorithm~\ref{alg:matroid_D}.
Assume that Algorithm~\ref{alg:matroid_D} halts when $s = k$. 
Furthermore, we assume that $e \in T_{\ell}$. 
For each integer $s \in [k]$, 
we define $F_s \coloneqq \bigcup_{j=1}^{s} T_j$.  
Define $F_0 \coloneqq \emptyset$.
Define $C \coloneqq \mathbb{C}_{{\cal M}_D\langle F \rangle}(e,B)$. 
Then $C$ is a circuit of ${\cal M}_{\ell}$. 
Thus, (i) holds. 

Let $f$ be an element in $C$.
Then $f \in T_{\ell}$ and 
$(B \cap T_{\ell}) + e - f$ is an independent set of ${\cal M}_{\ell}$.
Thus, since $B$ is a base of ${\cal M}_D\langle F \rangle$, 
$B + e - f$ is a base of ${\cal M}_D\langle F \rangle$. 
This and Lemma~\ref{lemma:base_matroid_D} imply that 
$B + e - f \in {\cal I}_D$. 
Thus, $f \in \mathbb{C}_D(e,B)$. 

Let $f$ be an element in $\mathbb{C}_D(e,B) \setminus C$. 
Assume that $e \succsim_D f$. 
Define $B^{\prime} \coloneqq B + e - f$. 
Then since $f \in \mathbb{C}_D(e,B)$, 
$B^{\prime} \in {\cal I}_D$.
Since $e \succsim_D f$ (i.e., $f \notin F_{\ell-1}$), 
$B \cap F_{\ell-1} = B^{\prime} \cap F_{\ell-1}$.
Thus, since 
$B^{\prime} \cap F_{\ell} \in {\cal I}_D$ and 
Lemma~\ref{lemma:base_matroid_D} implies that 
$B \cap F_{\ell-1}$ is a base of ${\cal M}_D|F_{\ell-1}$,
Lemma~\ref{lemma:contraction} implies that  
$B^{\prime} \cap T_{\ell}$ is an independent set of ${\cal M}_{\ell}$. 
However, since 
$C \subseteq B^{\prime} \cap T_{\ell}$, 
this is a contradiction. 
\end{proof}

\begin{lemma} \label{lemma:tail_D}
Let $F$ be a subset of $E$.
Let $C$ be a circuit of ${\cal M}_D$
such that 
$C \subseteq F$. 
Then ${\bf ta}_D(C)$
is a dependent set of 
${\cal M}_D\langle F \rangle$.
\end{lemma}
\begin{proof}
For notational simplicity, we define $X \coloneqq {\bf ta}_D(C)$. 
In what follows, we use the notation 
in Algorithm~\ref{alg:matroid_D}.
Assume that Algorithm~\ref{alg:matroid_D} halts when $s = k$, and 
$X \subseteq T_{\ell}$. 
Define $F_s\coloneqq \bigcup_{j=1}^{s} T_j$
for each integer $s \in [k]$. 
Define $F_0 \coloneqq \emptyset$. 
Let $B$ be a base of ${\cal M}\langle F \rangle$. 
Define $B_{\ell-1} \coloneqq B \cap F_{\ell-1}$. 

The definition of ${\cal M}_D\langle F \rangle$
implies that our goal is to prove that 
$X$ is a dependent set of 
${\cal M}_{\ell}$.
Since ${\cal M}_{\ell} = ({\cal M}_D/F_{\ell-1})|T_{\ell}$, 
it is sufficient to prove that 
$X$ is a dependent set of 
${\cal M}_D/F_{\ell-1}$.

Lemma~\ref{lemma:base_matroid_D}
implies that $B_{\ell-1}$ 
is a base of ${\cal M}_D|F_{\ell-1}$. 
We prove that there exists a circuit $C^{\prime}$ of 
${\cal M}_D$ such that 
$C^{\prime} \cap T_{\ell} \subseteq X$
and 
$C^{\prime} \setminus X \subseteq B_{\ell-1}$. 
If we can prove this, then 
$(C^{\prime} \cap T_{\ell}) \cup B_{\ell-1} \notin {\cal I}_D$ 
since $C^{\prime} \subseteq (C^{\prime} \cap T_{\ell}) \cup B_{\ell-1}$.
Thus, it follows from Lemma~\ref{lemma:contraction} that 
$C^{\prime} \cap T_{\ell}$ 
is a dependent set of ${\cal M}_D / F_{\ell-1}$.  
Since $C^{\prime} \cap T_{\ell} \subseteq X$, 
(I1) implies that $X$ is a dependent set of 
${\cal M}_D / F_{\ell-1}$. 

If $C \setminus X \subseteq B_{\ell-1}$, then 
since $C$ satisfies the above conditions, 
the proof is done. 
Assume that $C \setminus X \not\subseteq B_{\ell-1}$.
Then we prove that, for every element 
$f \in C \setminus (X \cup B_{\ell-1})$, 
there exists a circuit $C_f$ of ${\cal M}_D$ such that  
$f \in C_f$ and $C_f - f \subseteq B_{\ell-1}$.
Since $f \notin X$ and $C_f - f\subseteq B_{\ell-1}$, 
$X \cap C_f = \emptyset$.
Since $X \neq \emptyset$,
we can prove 
the existence of a circuit $C^{\prime}$ of 
${\cal M}_D$ satisfying the above conditions
by selecting an arbitrary element in $X$ and 
applying Lemma~\ref{lemma:circuit_union}.

Let $f$ be an element in $C \setminus (X \cup B_{\ell-1})$.
Since $C \setminus X \subseteq F_{\ell-1}$, 
there 
exists an integer 
$s \in [\ell-1]$ such that 
$f \in T_s$.
Thus, since Lemma~\ref{lemma:base_matroid_D}
implies that $B \cap F_s$ is a base of ${\cal M}_D|F_s$, 
there exists a circuit $C_f$ of 
${\cal M}_D$ such that 
$f \in C_f$ and $C_f - f \subseteq B \cap F_s \subseteq B_{\ell-1}$. 
\end{proof} 

\begin{lemma} \label{lemma:tail_circuit_D}
Let $F$ be a subset of $E$, and 
let $C$ be a circuit of ${\cal M}_D$
such that 
$C \subseteq F$. 
Then for every element $e \in {\bf ta}_D(C)$, 
there exists a circuit $C^{\prime}$ 
of ${\cal M}_D\langle F \rangle$ such that 
$e \in C^{\prime} \subseteq {\bf ta}_D(C)$. 
\end{lemma}
\begin{proof}
In this proof, we use the notation 
in Algorithm~\ref{alg:matroid_D}.
Assume that Algorithm~\ref{alg:matroid_D} halts when $s = k$, and 
${\bf ta}_D(C) \subseteq T_{\ell}$. 
Define $F_s \coloneqq \bigcup_{j=1}^s T_j$
for each integer $s \in [k]$. 
Define $F_0 \coloneqq \emptyset$. 
Let $B$ be a base of ${\cal M}\langle F \rangle$. 
Define $B_{\ell-1} \coloneqq B \cap F_{\ell-1}$. 

Let $e$ be an element in ${\bf ta}_D(C)$. 
Assume that there does not exist a circuit $C^{\prime}$ 
of ${\cal M}_{\ell}$ such that 
$e \in C^{\prime} \subseteq {\bf ta}_D(C)$. 
Define $C^{\ast}$ as a circuit of ${\cal M}_D$ satisfying the 
following conditions. 
\begin{description}
\item[(A1)]
$C^{\ast} \subseteq F$ and $e \in {\bf ta}_D(C^{\ast}) \subseteq {\bf ta}_D(C)$.
\item[(A2)]
$C^{\ast}$ minimizes 
$|{\bf ta}_D(C^{\ast})|$ 
among all the circuits of ${\cal M}_D$ satisfying (A1).
\end{description}
Notice that since $C$ satisfies (A1),
$C^{\ast}$ is well-defined. 
In what follows, we prove that 
there exists a circuit $C^{\circ}$ of ${\cal M}_D$ such that 
$C^{\circ}$ satisfies (A1) and 
$|{\bf ta}_D(C^{\circ})| < |{\bf ta}_D(C^{\ast})|$.
This contradicts the definition of $C^{\ast}$. 

For notational simplicity, we define $X \coloneqq {\bf ta}_D(C^{\ast})$. 
Lemma~\ref{lemma:tail_D} implies that 
$X$ is a dependent set of 
${\cal M}_{\ell}$. 
Thus, there exists a circuit $C^{\circ}$ of 
${\cal M}_{\ell}$ such that
$C^{\circ} \subseteq X$. 
Since $C^{\circ} \subseteq X \subseteq {\bf ta}_D(C)$, 
the above assumption implies that 
$e \notin C^{\circ}$. 
Since $C^{\circ}$ is a dependent set of 
${\cal M}_{\ell}$, 
Lemma~\ref{lemma:contraction} implies that 
$C^{\circ} \cup B_{\ell-1} \notin {\cal I}_D$. 
Thus, there exists a circuit 
$C^{\bullet}$ of ${\cal M}_D$ 
such that $C^{\bullet} \subseteq C^{\circ} \cup B_{\ell-1} \subseteq F$. 
Since $B_{\ell-1} \in {\cal I}_D$, 
$C^{\bullet}$ contains an element 
$f \in C^{\circ} \subseteq X \subseteq C^{\ast}$. 
Since 
$e \notin C^{\circ}$, 
$e \notin C^{\bullet}$. 
Thus, 
since $e \in C^{\ast}$, 
Lemma~\ref{lemma:elimination} implies that 
there exists a circuit $C^{\diamond}$ of ${\cal M}_D$
such that 
$e \in C^{\diamond} \subseteq (C^{\ast} \cup C^{\bullet}) - f$. 
Since 
${\bf ta}_D(C^{\bullet}) \subseteq C^{\circ} \subseteq X$, 
$e \in {\bf ta}_D(C^{\diamond}) \subseteq X - f$. 
Thus, $C^{\diamond}$ satisfies (A1) and 
$|{\bf ta}_D(C^{\diamond})| < |X|$. 
However, this contradicts the definition of 
$C^{\ast}$. 
\end{proof} 

\begin{algorithm}[ht]
Define $N_0 \coloneqq F$ and ${\cal M}_0^{\prime} \coloneqq {\cal M}_H$.
Set $s \coloneqq 0$.\\
If $N_0 = \emptyset$, then output the matroid $(\emptyset,\{\emptyset\})$ as 
${\cal M}_H\langle F \rangle$ and halt.\\
\While{$N_s\neq \emptyset$}
{
   \For{$j = 1,2$}
   {
      Set $s \coloneqq s+1$.\\
      Define $T_s \coloneqq {\bf he}_H(N_{s-1}) \cap E_j$ 
      and $N_s \coloneqq N_{s-1} \setminus T_s$.\\
      Define ${\cal M}_s \coloneqq {\cal M}_{s-1}^{\prime}|T_s$ 
      and ${\cal M}^{\prime}_t \coloneqq {\cal M}_{t-1}^{\prime} / T_t$.\\
   }
}
Output the direct sum of 
${\cal M}_1, {\cal M}_2, \dots, {\cal M}_s$ as ${\cal M}_H\langle F \rangle$ and halt.
\caption{Algorithm for defining ${\cal M}_H\langle F \rangle$}
\label{alg:matroid_H}
\end{algorithm}

\begin{lemma} \label{lemma:base_matroid_H}
Let $F$ be a subset of $E$, and 
let $B$ be a base of ${\cal M}_H\langle F \rangle$. 
Assume that Algorithm~\ref{alg:matroid_H} halts when $s = k$.
We use the notation in Algorithm~\ref{alg:matroid_H}. 
Define $F_s \coloneqq \bigcup_{j = 1}^sT_j$ 
and $B_s \coloneqq B \cap F_s$ 
for each integer $s \in [k]$. 
Then $B_s$ is a base of ${\cal M}_H|F_s$
for every integer $s \in [k]$.
\end{lemma}
\begin{proof}
This lemma can be proved in the same way as Lemma~\ref{lemma:base_matroid_D}.  
\end{proof} 

Lemma~\ref{lemma:base_matroid_H}
implies that, for every subset $F \subseteq E$, 
every base $B$ of ${\cal M}_H\langle F \rangle$, and 
every element $e \in F \setminus B$, 
we have $e \in {\bf cl}_H(B)$.

\begin{lemma} \label{lemma:circuit_matroid_H}
Let $F$ be a subset of $E$.
Let $B$ be a base of ${\cal M}_H\langle F \rangle$.
Let $e$ be an element in $F \setminus B$. 
\begin{description}
\item[(i)]
$e \sim_H f$ for every element 
$f \in \mathbb{C}_{{\cal M}_H\langle F \rangle}(e,B)$.
\item[(ii)]
$\mathbb{C}_{{\cal M}_H\langle F \rangle}(e,B) \subseteq 
\mathbb{C}_H(e,B)$. 
\item[(iii)]
For every element 
$f \in \mathbb{C}_H(e,B) \setminus \mathbb{C}_{{\cal M}_H\langle F \rangle}(e,B)$, 
if $e \in E_1$ or $e,f \in E_2$, then 
$f \succ_H e$.
\end{description}
\end{lemma} 
\begin{proof}
In this proof, we use the notation in Algorithm~\ref{alg:matroid_H}.
Assume that Algorithm~\ref{alg:matroid_H} halts when $s = k$. 
Furthermore, we assume that $e \in T_{\ell}$. 
For each integer $s \in [k]$, 
we define $F_s \coloneqq \bigcup_{j=1}^{s} T_j$.  
Define $F_0 \coloneqq \emptyset$.
Define $C \coloneqq \mathbb{C}_{{\cal M}_H\langle F \rangle}(e,B)$. 
Then $C$ is a circuit of ${\cal M}_{\ell}$. 
Thus, (i) holds. 

Let $f$ be an element in $C$.
Then $f \in T_{\ell}$ and 
$(B \cap T_{\ell}) + e - f$ is an independent set of ${\cal M}_{\ell}$.
Thus, $B + e - f$ is a base of ${\cal M}_H\langle F \rangle$. 
This and Lemma~\ref{lemma:base_matroid_H} imply that 
$B + e - f \in {\cal I}_H$. 
Thus, $f \in \mathbb{C}_H(e,B)$. 

Let $f$ be an element in $\mathbb{C}_H(e,B) \setminus C$. 
Assume that $e \succsim_H f$. 

\begin{claim} \label{claim_1:lemma:circuit_matroid_H}
If $e \in E_1$ or $e,f \in E_2$, then 
$f \notin F_{\ell - 1}$. 
\end{claim}
\begin{proof}
If $e \in E_1$, then $g \succ_H e$ for every 
element $g \in F_{\ell-1}$. 
Thus, $f \notin F_{\ell-1}$. 

If $e \in E_2$ and $f \in F_{\ell-1}$, then 
at least one of $f \in E_1$, 
$f \succ_H e$ holds. 
Thus, since $e \succsim_H f$, 
if $f \in E_2$, then $f \notin F_{\ell -1}$. 
\end{proof} 

Assume that
$e \in E_1$ or $e,f \in E_2$. 
Define $B^{\prime} \coloneqq B + e - f$. 
Since $f \in \mathbb{C}_H(e,B)$, 
$B^{\prime} \in {\cal I}_H$.
In addition, since $f \notin F_{\ell-1}$ follows from 
Claim~\ref{claim_1:lemma:circuit_matroid_H}, 
$B \cap F_{\ell-1} = B^{\prime} \cap F_{\ell-1}$.
Lemma~\ref{lemma:base_matroid_H} implies that 
$B \cap F_{\ell-1}$ is a base of ${\cal M}_H|F_{\ell-1}$. 
Thus, 
$B^{\prime} \cap T_{\ell}$ is an independent set of ${\cal M}_{\ell}$. 
Since 
$C \subseteq B^{\prime} \cap T_{\ell}$, 
this is a contradiction. 
\end{proof}

\begin{algorithm}[ht]
Define $N_0 \coloneqq F$, $T_0^{\ast} \coloneqq \emptyset$, 
and ${\cal M}_0^{\prime} \coloneqq {\cal M}_D$.
Set $s \coloneqq 0$.\\
\While{$N_s \neq \emptyset$}
{Set $s \coloneqq s+1$.\\
Define $T_s \coloneqq {\bf he}_D(N_{s-1})$, 
$N_s \coloneqq N_{s-1} \setminus T_s$, and
${\cal M}_s^{\prime} \coloneqq {\cal M}_{s-1}^{\prime} / T_s$.\\
Define $T_s^{\ast} \coloneqq T_{s-1}^{\ast} \cup \{e \in T_s \mid \mbox{$\{e\}$ 
is an independent set of ${\cal M}_{s-1}^{\prime}$}\}$.
}
Output $T_s^{\ast}$ as ${\rm Ch}_D(F)$ and halt.
\caption{Algorithm for defining ${\rm Ch}_D(F)$}
\label{alg:choice_D}
\end{algorithm}

\begin{lemma} \label{lemma:strong_dominance_choice_D}
Let $F$ be a subset of $E$, and 
let $e$ be an element in $F$. 
Assume that there exists a circuit $C$ of ${\cal M}_D$
such that 
$e \in C \subseteq F$ and 
$f \succ_D e$ for every element $f \in C - e$. 
Then we have $e \notin {\rm Ch}_D(F)$. 
\end{lemma}
\begin{proof}
In this proof, we use the notation in Algorithm~\ref{alg:choice_D}.
Assume that $e \in T_{\ell}$. 
Since ${\bf ta}_D(C) = \{e\}$, 
Lemma~\ref{lemma:tail_D}
implies that 
$\{e\}$ is not a dependent set of ${\cal M}\langle F \rangle$. 
Thus, $\{e\}$ is a dependent set of ${\cal M}_{\ell-1}^{\prime}|T_{\ell}$.
This implies that 
since 
$\{e\}$ is a dependent set of ${\cal M}_{\ell-1}^{\prime}$, 
$e \notin T_{\ell}^{\ast} \setminus T_{\ell-1}^{\ast}$. 
\end{proof} 

\begin{lemma} \label{lemma:outside_choice_D}
Let $F$ be a subset of $E$, and 
let $e$ be an element in $F$. 
Assume that $e \notin {\rm Ch}_D(F)$. 
Then 
there exists a circuit $C$ of ${\cal M}_D$
such that 
$e \in C$, 
$C - e \subseteq {\rm Ch}_D(F)$, and 
$f \succ_D e$ for every element $f \in C - e$. 
\end{lemma}
\begin{proof}
In this proof, we use the notation in Algorithm~\ref{alg:choice_D}.
Assume that $e \in T_{\ell}$. 
Since $e \notin {\rm Ch}_D(F)$ 
$\{e\}$ is a dependent set of ${\cal M}_{\ell-1}^{\prime}$. 
Since $\{e\} \in {\cal I}_D$, $\ell > 1$ holds. 
Define $F_{\ell-1} \coloneqq \bigcup_{j=1}^{\ell-1}T_j$. 
Lemma~\ref{lemma:minor} implies that 
${\cal M}_{\ell-1}^{\prime} = {\cal M}_D/F_{\ell-1}$. 
Let $B$ be a base of ${\cal M}_D \langle F \rangle$. 
Lemma~\ref{lemma:base_matroid_D}
implies that 
$B \cap F_s$ is a base of 
${\cal M}_D|F_s$ 
for every integer $s \in [\ell-1]$. 
Define $B^{\prime} \coloneqq B \cap F_{\ell-1}$. 

\begin{claim} \label{claim_1:lemma:outside_choice_D} 
$B^{\prime} \subseteq T_{\ell-1}^{\ast}$. 
\end{claim} 
\begin{proof}
Assume that $B^{\prime} \not\subseteq T_{\ell-1}^{\ast}$. 
Let $f$ be an element in $B^{\prime} \setminus T_{\ell-1}^{\ast}$. 
Since $B^{\prime} \subseteq F_{\ell-1}$, 
there exists an integer $j \in [\ell-1]$ such that 
$f \in T_j$.
Then $\{f\}$ is a dependent set of ${\cal M}_{j-1}^{\prime}$. 
Thus,  
Lemma~\ref{lemma:contraction} implies that
$(B \cap F_{j-1}) + f \notin {\cal I}_D$. 
However, since $(B \cap F_{j-1}) + f \subseteq B$, this contradicts the fact that 
$B \in {\cal I}_D$. This completes the proof. 
\end{proof} 

Since $\{e\}$ is a dependent set of ${\cal M}_{\ell-1}^{\prime}$, 
Lemma~\ref{lemma:contraction} implies that
$B^{\prime} + e \notin {\cal I}_D$.
Thus, 
Lemma~\ref{lemma:closure} implies that
$e \in {\bf cl}_D(B^{\prime})$.
Define $C \coloneqq \mathbb{C}_D(e, B^{\prime})$.
Then since 
Claim~\ref{claim_1:lemma:outside_choice_D}
implies that 
$B^{\prime} \subseteq {\rm Ch}_D(F)$, 
$C - e \subseteq {\rm Ch}_D(F)$.
Furthermore, since  
$C - e \subseteq B^{\prime} \subseteq F_{\ell-1}$ and 
$e \in T_{\ell}$, 
$f \succ_D e$ for every element $f \in C - e$. 
This completes the proof. 
\end{proof} 

\begin{lemma} \label{lemma:rank_equal} 
Let $F$ be a subset of $E$.
Define $K \coloneqq {\rm Ch}_D(F)$. 
Then 
${\bf rk}_D(F) = {\bf rk}_{{\cal M}_D\langle K \rangle}(K)$. 
\end{lemma} 
\begin{proof}
Notice that Lemma~\ref{lemma:outside_choice_D} 
implies that, for every element $e \in F \setminus K$, 
there exists a circuit $C$ of ${\cal M}_D$ such that 
$e \in C \subseteq K + e$.
Thus, 
since Lemma~\ref{lemma:closure} 
implies that 
$F \setminus K \subseteq {\bf cl}_D(K)$, 
\begin{equation*}
{\bf rk}_D(K) = {\bf rk}_D(K \cup (F \setminus K)) = {\bf rk}_D(F). 
\end{equation*}
Since Lemma~\ref{lemma:base_matroid_D} implies that 
${\bf rk}_D(K) = {\bf rk}_{{\cal M}_D\langle K \rangle}(K)$, 
this completes the proof. 
\end{proof} 

\begin{lemma} \label{lemma:outside_circuit}
Let $F$ be a subset of $E$.
Define $K \coloneqq {\rm Ch}_D(F)$. 
Assume that there exists a circuit $C$ of ${\cal M}_D$
such that $C \subseteq F$ and ${\bf ta}_D(C) \cap K \neq \emptyset$. 
Let $e$ be an element in ${\bf ta}_D(C) \cap K$. 
Then there exists a circuit $C^{\prime}$ of ${\cal M}_D$ 
such that $C^{\prime} \subseteq K$ and $e \in {\bf ta}_D(C^{\prime}) \subseteq {\bf ta}_D(C) \cap K$.
\end{lemma}
\begin{proof}
If $C \subseteq K$, then the proof is done. 
Thus, we assume that $C \not\subseteq K$. 
Then it follows from Lemma~\ref{lemma:outside_choice_D}
that, 
for 
every element $f \in C \setminus K$, 
there exists a circuit $C_f$ of ${\cal M}_D$
such that 
$f \in C_f$, 
$C_f - f \subseteq K$, and 
$g \succ_D f$ for every element $g \in C_f - f$. 
Thus, since 
$e \notin C_f$ for 
every element $f \in C \setminus K$, 
Lemma~\ref{lemma:circuit_union}
implies that 
there exists a circuit $C^{\prime}$ of ${\cal M}_D$ 
satisfying the conditions in this lemma.
\end{proof} 

\subsection{Main algorithm} 

Our main algorithm is described in Algorithm~\ref{alg:main}. 

\begin{algorithm}[ht]
Set $t \coloneqq 0$. Define $R_0 \coloneqq \emptyset$\\
\Do{$R_t \neq P_{t,i}$}
{
    Set $t \coloneqq t+1$ and $i \coloneqq 0$. 
    Define $P_{t,0} \coloneqq R_{t-1}$.\\
    \Do{$P_{t,i} \neq P_{t,i-1}$}
    {
        Set $i \coloneqq i + 1$.\\
        Define $K_{t,i} \coloneqq {\rm Ch}_D(E \setminus P_{t,i-1})$.
        Define $Q_{t,i}$ as a base of ${\cal M}_H\langle K_{t,i} \rangle$.\\
        \If{${\bf rk}_D(E \setminus P_{t,i-1}) < {\bf rk}_H(K_{t,i})$}
        {
            Define $P_{t,i} \coloneqq P_{t,i-1}$ and $R_{t} \coloneqq P_{t,i}$.
            Output ${\bf No}$, and halt.
        }
        \uIf{$Q_{t,i} \cap E_1 \neq K_{t,i} \cap E_1$}
        {
            Define $P_{t,i} \coloneqq P_{t,i-1} \cup ((K_{t,i} \setminus Q_{t,i}) \cap E_1)$.
        }
        \Else
        {
            Find a maximum-size common independent set $I_{t,i}$ of 
            ${\cal M}_D\langle K_{t,i} \rangle, {\cal M}_H\langle K_{t,i} \rangle$.\\
            \uIf{$|I_{t,i}| < {\bf rk}_D(E \setminus P_{t,i-1})$}
            {
                Find the critical subset $Z_{t,i}$ 
                of ${\cal M}_D\langle K_{t,i} \rangle, {\cal M}_H\langle K_{t,i} \rangle$.\\
                Define $P_{t,i} \coloneqq P_{t,i-1} \cup Z_{t,i}$.
            }
            \Else
            {
                Define $P_{t,i} \coloneqq P_{t,i-1}$.
            }
        }
    }
    \uIf{$P_{t,i} \cap {\bf block}(I_{t,i}) \neq \emptyset$}
    {
        Define $b_t$ as an element in $P_{t,i} \cap {\bf block}(I_{t,i})$.
        Define $R_t \coloneqq P_{t,i} \cup {\bf ta}_H(\mathbb{C}_H(b_t, I_{t,i}))$.
    }
    \Else
    {
        Define $R_t \coloneqq P_{t,i}$.
    }
}
\uIf{there exists an element $e_{\rm R} \in R_t$ such that $I_{t,i} + e_{\rm R} \in {\cal I}_H$}
{
    Output ${\bf No}$, and halt.
}
\Else
{
    Output $I_{t,i}$, and halt.
}
\caption{Proposed algorithm}
\label{alg:main}
\end{algorithm}

In the rest of this subsection, we use the notation in Algorithm~\ref{alg:main}. 
Notice that since $P_{t,i-1} \subseteq P_{t,i}$ and 
$P_{t,i} \subseteq R_t$, 
Algorithm~\ref{alg:main} is a polynomial-time algorithm. 

\begin{lemma} \label{lemma:Z_nonempty} 
In Step~22 of Algorithm~\ref{alg:main}, 
$|I_{t,i}| \ge {\bf rk}_D(E \setminus P_{t,i-1})$. 
\end{lemma}
\begin{proof}
It suffices to prove that,
in Step~15 of Algorithm~\ref{alg:main},  
$Z_{t,i} \neq \emptyset$. 
In Step~15 of Algorithm~\ref{alg:main}.
Lemma~\ref{lemma:rank_equal} implies that 
$|I_{t,i}| < {\bf rk}_D(E \setminus P_{t,i-1}) = 
{\bf rk}_{{\cal M}_D\langle K_{t,i} \rangle}(K_{t,i})$. 
Thus, Lemma~\ref{lemma:Edmonds70} implies that 
$Z_{t,i} \neq \emptyset$.
This completes the proof. 
\end{proof}

\begin{lemma} \label{lemma:main_alg:base_D}
In Step~22 of Algorithm~\ref{alg:main}, 
$I_{t,i}$ is a base of ${\cal M}_D\langle K_{t,i} \rangle$ and 
${\cal M}_D|(E \setminus P_{t,i})$. 
\end{lemma} 
\begin{proof}
For notational simplicity, 
we define 
\begin{equation*}
I \coloneqq I_{t,i}, \ \ 
P \coloneqq P_{t,i} = P_{t,i-1}, \ \ 
K \coloneqq K_{t,i}, \ \ 
{\cal K}_D \coloneqq {\cal M}_D\langle K\rangle. 
\end{equation*} 
Lemma~\ref{lemma:Z_nonempty} 
implies that 
$|I| \ge {\bf rk}_D(E \setminus P)$. 
Furthermore, $I \subseteq K \subseteq E \setminus P$.
Thus, 
if we can prove that 
$I \in {\cal I}_D$, then 
$I$ is a base of ${\cal M}_D|(E \setminus P)$. 
This and Lemma~\ref{lemma:base_matroid_D} 
imply 
that 
it is sufficient to prove that 
$I$ is a base of ${\cal K}_D$. 
To this end, since $I$ is an independent set of ${\cal K}_D$,
it is sufficient to prove that 
$|I| \ge {\bf rk}_{{\cal K}_D}(K)$. 
Lemma~\ref{lemma:base_matroid_D} 
implies 
that 
${\bf rk}_{{\cal K}_D}(K) = {\bf rk}_D(K)$. 
Furthermore, since 
$K \subseteq E \setminus P$, 
${\bf rk}_D(E\setminus P) \ge {\bf rk}_D(K)$.
Since 
$|I| \ge {\bf rk}_D(E \setminus P)$,
this completes the proof. 
\end{proof} 

\begin{lemma} \label{lemma:main_alg:base_H}
In Step~22 of Algorithm~\ref{alg:main}, 
$I_{t,i}$ is a base of ${\cal M}_H\langle K_{t,i} \rangle$.
\end{lemma} 
\begin{proof}
Lemma~\ref{lemma:Z_nonempty} and 
the condition in Step~7 of Algorithm~\ref{alg:main} 
imply that 
\begin{equation*}
{\bf rk}_H(K_{t,i}) =
{\bf rk}_{{\cal M}_H\langle K_{t,i} \rangle}(K_{t,i})
\ge 
|I_{t,i}| \ge
{\bf rk}_D(E \setminus P_{t,i-1}) \ge {\bf rk}_H(K_{t,i}), 
\end{equation*}
where 
the equation follows from Lemma~\ref{lemma:base_matroid_D} and 
the first inequality follows from the fact that 
$I_{t,i}$ is an independent set of 
${\cal M}_H\langle K_{t,i} \rangle$. 
Then 
the first inequality 
holds with equality. 
Thus, $I_{t,i}$ is a base of 
${\cal M}_H\langle K_{t,i} \rangle$. 
This completes the proof. 
\end{proof} 

\begin{lemma} \label{lemma:full_E_1}
In Step~22 of Algorithm~\ref{alg:main}, 
$K_{t,i} \cap E_1 \subseteq I_{t,i}$.
\end{lemma} 
\begin{proof}
In this proof, we use the notation  
in Algorithm~\ref{alg:matroid_H} by replacing $F$ with $K_{t,i}$.
Assume that Algorithm~\ref{alg:matroid_H} halts when $s = k$. 
Then since the condition of Step~10 of Algorithm~\ref{alg:main}
implies that 
every element in $K_{t,i} \cap E_1$ is contained in some 
base of ${\cal M}_H\langle K_{t,i} \rangle$,
$T_s$ is the unique base of 
${\cal M}_s$ for every integer $s \in [k]$ 
such that $s$ is odd. 
Furthermore, 
since Lemma~\ref{lemma:main_alg:base_H}
implies that $I_{t,i}$ is a base of 
${\cal M}_H\langle K_{t,i} \rangle$, 
$I_{t,i} \cap T_s$ is a base of 
${\cal M}_s$ 
for every integer $s \in [k]$
such that $s$ is odd. 
Thus, 
$I_{t,i} \cap T_s = T_s$
for every integer $s \in [k]$. 
This completes the proof. 
\end{proof}

\begin{lemma} \label{lemma:key}
Assume that Algorithm~\ref{alg:main} halts when $t = k$. 
If there exists a non-uniformly stable common 
independent set $J$ of ${\cal M}_D, {\cal M}_H$, 
then we have $J \cap R_k = \emptyset$. 
\end{lemma}

We prove Lemma~\ref{lemma:key} in 
Section~\ref{section:proof_key_lemma}. 

\section{Correctness} 

Assume that Algorithm~\ref{alg:main} halts when $t = k$. 
In addition, we assume that $P_{t,i_t} = P_{t,i_t-1}$
for each integer $t \in [k]$.
Namely,
we assume that
Algorithm~\ref{alg:main} proceeds to 
Step~22 when $i = i_t$ 
in the $t$th iteration of 
Steps~2 to 27. 
For notational simplicity, 
we define 
\begin{equation*}
I \coloneqq I_{k,i_k}, \ \ 
R \coloneqq R_k, \ \ 
K \coloneqq K_{k,i_k}, \ \ 
{\cal K}_D \coloneqq {\cal M}_D\langle K\rangle, \ \ 
{\cal K}_H \coloneqq {\cal M}_H\langle K\rangle.
\end{equation*} 
Notice that 
$R = P_{k,i_k} = P_{k,i_k-1}$. 

\begin{lemma} \label{lemma:stability}
If Algorithm~\ref{alg:main} halts at Step~31, 
then the output of Algorithm~\ref{alg:main}
is a non-uniformly stable common 
independent set of ${\cal M}_D, {\cal M}_H$. 
\end{lemma}
\begin{proof}
Assume that Algorithm~\ref{alg:main} halts at Step~31.
Let $e$ be an element in $E \setminus I$, and 
we prove that $e$ does not block $I$. 

First, we assume that $e \in R$.
Assume that $e$ blocks $I$.
Then we derive a contradiction. 
Since Algorithm~\ref{alg:main} does not halt at Step~29, 
$I + e \notin \mathcal{I}_H$. 
Furthermore, Lemmas~\ref{lemma:base_matroid_H}
and \ref{lemma:main_alg:base_H}
imply that $I$ is an independent set of ${\cal M}_H$.
Thus,  
$e \in {\bf cl}_H(I)$. 
This implies that 
since $e \in R$ and $e$ blocks $I$, 
$e \in R \cap {\bf block}(I)$, i.e., 
$R \cap {\bf block}(I) \neq \emptyset$. 
Since we assume that $\{f\} \in {\cal I}_H$ for 
every element $f \in E$, 
we have ${\bf ta}_H(\mathbb{C}_H(b_k, I)) \neq \{b_k\}$. 
However, this contradicts the fact that 
$R = P_{k,i_k}$. 
 
From here, we assume that $e \in E \setminus R$.
First, we assume that $e \in K$. 
Then 
Lemmas~\ref{lemma:main_alg:base_D} and 
\ref{lemma:main_alg:base_H} imply that
$e \in {\bf cl}_{{\cal K}_S}(I)$
for every 
element $S \in \{D,H\}$. 
In addition, 
Lemmas~\ref{lemma:circuit_matroid_D} 
and 
\ref{lemma:circuit_matroid_H}
imply
that, for every element $S \in \{D,H\}$, 
$f \succsim_S e$ for every 
element $f \in \mathbb{C}_S(e,I)$. 
Notice that 
since $e \notin I$, 
Lemma~\ref{lemma:full_E_1}
implies that 
$e \in E_2$. 
Thus, $e$ does not block $I$. 

Next, we assume that $e \notin K$. 
Since $e \in E \setminus R$,  
Lemmas~\ref{lemma:closure} and
\ref{lemma:main_alg:base_D}
imply that $e \in {\bf cl}_D(I)$. 

\begin{claim} \label{claim_1:lemma:stability}
$f \succ_D e$ for every element $f \in \mathbb{D}_D(e,I)$. 
\end{claim}
\begin{proof}
We prove that 
there exists 
a circuit $C^{\prime}$ of ${\cal M}_D$ such that 
$e \in C^{\prime} \subseteq I + e$ 
and 
$f \succ_D e$ for every element $f \in C^{\prime} - e$.
Then 
since $C^{\prime} = \mathbb{C}_D(e,I)$, 
this completes the proof. 

Since $e \notin K$, 
it follows from Lemma~\ref{lemma:outside_choice_D}
that there exists a circuit $C$ of ${\cal M}_D$ 
such that $e \in C$, $C - e \subseteq K$, and 
$f \succ_D e$ for every element $f \in C - e$. 
Thus, 
it suffices to prove that, for every 
element $f \in C \setminus (I+e)$, there exists a circuit 
$C_f$ of ${\cal M}_D$ such that 
$f \in C_f \subseteq I + f$ and 
$g \succsim_D f$ for every element $g \in C_f$.
If this is true, then
for every element $f \in C \setminus (I + e)$,  
since 
$f \succ_D e$, $e \notin C_f$. 
Thus, Lemma~\ref{lemma:circuit_union} completes the proof of 
this claim.

Let $f$ be an element in $C \setminus (I+e)$. 
Since $f \in K$, 
Lemma~\ref{lemma:main_alg:base_D}
implies that 
$f \in {\bf cl}_{{\cal K}_D}(I)$. 
Thus, 
Lemma~\ref{lemma:circuit_matroid_D}
implies that since $f \in \mathbb{C}_{{\cal K}_D}(f,I)$, 
$g \succsim_D f$
for every element  
$g \in \mathbb{C}_D(f,I)$. 
Thus, we can define $C_f$ as 
$\mathbb{C}_D(f,I)$. 
This completes the proof. 
\end{proof}

Claim~\ref{claim_1:lemma:stability}
implies that 
$e$ does not block $I$. 
This completes the proof of this case. 
\end{proof} 

\begin{lemma} \label{lemma:no}
If Algorithm~\ref{alg:main} halts at Step~8 or Step~29, 
then there does not exist 
a non-uniformly stable common 
independent set of ${\cal M}_D, {\cal M}_H$. 
\end{lemma}
\begin{proof}
Assume that Algorithm~\ref{alg:main} halts at Step~8 or Step~29.
Furthermore, we assume that 
there exists a non-uniformly stable common 
independent set $J$ of ${\cal M}_D, {\cal M}_H$.
Notice that Lemma~\ref{lemma:key} 
implies that 
$J \subseteq E \setminus R$. 
Define the subset $I^+ \subseteq E$ as follows. 
If 
Algorithm~\ref{alg:main} halts at Step~8, then 
we define $I^+$ as a base of ${\cal M}_H|K$.
Then since $J \subseteq E \setminus R$ and 
$J \in {\cal I}_D$, the condition in Step~7 of 
Algorithm~\ref{alg:main} 
implies that $|I^+| > |J|$. 
If 
Algorithm~\ref{alg:main} halts at Step~29, then 
we define $I^+ \coloneqq I + e_{\rm R}$. 
Since Lemma~\ref{lemma:main_alg:base_D} implies that 
$I$ is a base of ${\cal M}_D|(E \setminus R)$,
$|I^+| > |I| \ge |J|$. 

\begin{claim} \label{claim_1:lemma:no}
Let $e$ be an element in $I^+ \setminus J$
such that $e \in {\bf cl}_D(J)$.
Then 
there exists 
an element $f \in \mathbb{D}_D(e,J)$ such that 
$e \succsim_D f$. 
\end{claim}
\begin{proof}
Assume that 
$f \succ_D e$ for every 
element $f \in \mathbb{D}_D(e,J)$. 
Since $e \in I^+$, there exist integers 
$t \in [k]$ and $i \in [i_t]$ such that 
$e \in K_{t,i}$. 
Thus, $e \in E \setminus P_{t,i-1}$. 
Since $J \subseteq E \setminus R$
and 
$P_{t,i-1} \subseteq R$,
we have 
$J \subseteq E \setminus P_{t,i-1}$. 
This implies that $\mathbb{C}_D(e,J) \subseteq E \setminus P_{t,i-1}$. 
Thus, 
it follows from Lemma~\ref{lemma:strong_dominance_choice_D} that 
$e \notin K_{t,i}$.
However, this contradicts the fact that $e \in K_{t,i}$. 
\end{proof} 

Since $J$ is non-uniformly stable, it follows from 
Claim~\ref{claim_1:lemma:no} that, for every element $e \in I^+ \setminus J$, 
$e \in {\bf cl}_H(J)$ and 
$f \succsim_H e$ for every element $f \in \mathbb{D}_H(e,J)$. 

From here, we prove that there exists a circuit $C$ of 
${\cal M}_H$ such that $C \subseteq I^+$. 
The existence of such a circuit $C$ contradicts the fact that 
$I^+ \in {\cal I}_H$. 
To this end, we consider a pair $(X,Y)$ of subsets of $E$ 
satisfying the following conditions. 
\begin{description}
\item[(B1)]
$X \subseteq I^+ \setminus J$,  
$Y \subseteq J \setminus I^+$, and $|X| > |Y|$. 
\item[(B2)]
For every element $e \in X$, 
there exists a circuit $C_e$ of ${\cal M}_H$ such that 
$C_e \cap X = \{e\}$
and 
$C_e \subseteq I^+ \cup Y$. 
\end{description}
Notice that 
$(I^+ \setminus J, J \setminus I^+)$ 
satisfying the above conditions. 

Let $(X,Y)$ be a pair of subsets of $E$ such that
it satisfies (B1) and (B2) and minimizes $|Y|$ among all the 
pairs of subsets of $E$ 
satisfying (B1) and (B2).
In what follows, we prove that $|Y| = 0$. 
This implies that 
there exists a circuit $C$ of 
${\cal M}_H$ such that $C \subseteq I^+$. 

Assume that $Y \neq \emptyset$. 
For each element $e \in X$, 
we define $C_e$ as a circuit 
of ${\cal M}_H$ satisfying the conditions in (B2).
Then since $I^+ \in {\cal I}_H$, 
$C_e \cap Y \neq \emptyset$
for every element $e \in X$. 
Since $|X| > |Y|$, there exists
an element $e \in X$
such that 
$C_e \cap C_f \cap Y \neq \emptyset$
for some element $f \in X - e$. 
Fix such an element $e \in X$.
Let $g$ be an element in $Y$ such that 
$g \in C_e \cap C_f$
for some element $f \in X - e$. 
Define $W$ as the set of elements $f \in X - e$
such that $g \in C_f$. 
Then for each element $f \in W$, 
Lemma~\ref{lemma:elimination} implies that 
there exists a circuit $C_f^{\prime}$ of ${\cal M}_H$ 
such that 
$f \in C_f^{\prime} \subseteq I^+ \cup (Y - g)$ and 
$C_f^{\prime} \cap (X - e) = \{f\}$. 
Thus, 
$(X - e ,Y - g)$ satisfies 
(B1) and (B2).
However, since $|Y| > |Y - g|$, 
this contradicts 
the definition of 
$(X,Y)$. 
This completes the proof. 
\end{proof} 

We are now ready to prove the main result of this paper. 

\begin{theorem}
Algorithm~\ref{alg:main} can correctly check 
the existence of a non-uniformly stable common independent set 
of ${\cal M}_D,{\cal M}_H$, and 
find one if it exists. 
\end{theorem}
\begin{proof}
This theorem follows from 
Lemmas~\ref{lemma:stability} and \ref{lemma:no}. 
\end{proof} 

\section{Proof of Lemma~\ref{lemma:key}}
\label{section:proof_key_lemma}

In this section, we prove Lemma~\ref{lemma:key}.
In what follows, we use the notation in Algorithm~\ref{alg:main}. 
Assume that there exists a non-uniformly stable common 
independent set of ${\cal M}_D, {\cal M}_H$
containing an element in $R_k$. 
Then we call an element $e \in R_k$ a \emph{bad element} if 
there exists 
a non-uniformly stable common 
independent set of ${\cal M}_D, {\cal M}_H$
containing $e$. 
Notice that the above assumption implies that,
there exists an integer $t \in [k]$
such that $R_t\setminus R_{t-1}$ contains a bad element.
Define $z$ as the minimum integer among all the integers $t \in [k]$
such that  
there exists a bad element in $R_t\setminus R_{t-1}$. 

We divide the proof into the following two cases. 
\begin{description}
\item[Case~1.]
There exists an integer $i \in [i_z]$ such that 
$P_{z,i} \setminus P_{z,i-1}$ contains a bad element. 
\item[Case~2.]
For any integer $i \in [i_z]$, 
$P_{z,i} \setminus P_{z,i-1}$ does not contain a bad element. 
\end{description}

\subsection{Case~1}
 
Define $q$ as the minimum integer in $[i_z]$ 
such that $P_{z,q} \setminus P_{z,q-1}$ contains a bad element $e$.
Then there exists a non-uniformly stable common independent set 
$J$ of ${\cal M}_D,{\cal M}_H$ such that $e \in J$. 
The minimality of $q$ implies that 
$J \subseteq E \setminus P_{z,q-1}$. 
For notational simplicity, 
we define 
$K \coloneqq K_{z,q}$,
$Q \coloneqq Q_{z,q}$.
Furthermore, if $Q \cap E_1 = K \cap E_1$, then 
we define $Z \coloneqq Z_{z,q}$.

First, we assume that $e \in (K \setminus Q) \cap E_1$. 
Then since $Q$ is a base of ${\cal K}_H$, 
$e \in {\bf cl}_{{\cal K}_H}(Q)$. 
Define $C \coloneqq \mathbb{C}_H(e,Q)$. 
Then Lemma~\ref{lemma:circuit_matroid_H} implies that 
$f \succsim_H e$ for every element $f \in C$. 
We prove that, for every element $f \in C \setminus J$, 
there exists a circuit $C_f$ of ${\cal M}_H$ such that 
$f \in C_f$, 
$C_f - f \subseteq J$, and $e \notin C_f$. 
If we can prove this, then 
since $e \in C$, 
Lemma~\ref{lemma:circuit_union} implies that 
there exists a circuit $C^{\prime}$ of ${\cal M}_H$
such that $C^{\prime} \subseteq J$. 
However, 
this contradicts the fact that $J \in {\cal I}_H$. 

Let $f$ be an element in $C \setminus J$. 
Since $f \in C - e \subseteq K$, 
Lemma~\ref{lemma:strong_dominance_choice_D} 
implies that 
if $f \in {\bf cl}_D(J)$, then 
there exists 
an element $g \in \mathbb{D}_D(f,J)$ such that 
$f \succsim_D g$.
Thus, since $J$ is a non-uniformly stable common 
independent set of ${\cal M}_D, {\cal M}_H$, 
we have $f \in {\bf cl}_H(J)$ and 
$g \succsim_H f$ for every element $g \in \mathbb{C}_H(f,J)$. 
Define $C_f \coloneqq \mathbb{C}_H(f,J)$. 
What remains is to prove that $e \notin C_f$. 
If $f \succ_H e$, then 
since $g \succsim_H f$ for every element $g \in C_f$, 
$e \notin C_f$. 
Thus, we assume that $f \sim_H e$. 
Since $e \in E_1$, 
if 
$f \notin \mathbb{C}_{{\cal K}_H}(e,Q)$, 
then Lemma~\ref{lemma:circuit_matroid_H}
implies that $f \succ_H e$.
However, this contradicts the assumption that 
$f \sim_H e$. 
Thus, $f \in \mathbb{C}_{{\cal K}_H}(e,Q)$. 
This implies that since $e \in E_1$, $f \in E_1$. 
Thus, since $J$ is a non-uniformly stable common 
independent set of ${\cal M}_D, {\cal M}_H$,
we have 
$g \succ_H f$ for every element $g \in \mathbb{D}_H(f,J)$. 
This implies that 
$e \notin C_f$.
This completes the proof of this case. 

Next, we assume that $e \in Z$. 
Define 
\begin{equation*}
U \coloneqq Z \setminus 
{\bf cl}_{{\cal K}_D}(K \setminus (Z \setminus J)).
\end{equation*} 
Since $e \notin U$ follows from $e \in J$, we have $U \subsetneq Z$.
Thus, the minimality of $Z$ implies that   
\begin{equation} \label{eq_1:critical_set} 
{\bf rk}_{{\cal K}_D}(K \setminus U) 
+ {\bf rk}_{{\cal K}_H}(U)
>  
{\bf rk}_{{\cal K}_D}(K \setminus Z) 
+ {\bf rk}_{{\cal K}_H}(Z).
\end{equation}
Define $B$ as a base of ${\cal K}_H|U$. 
Then \eqref{eq_1:critical_set} implies that 
\begin{equation} \label{eq_2:critical_set} 
|B| > {\bf rk}_{{\cal K}_H}(Z) - 
{\bf rk}_{{\cal K}_D}(K \setminus U) +
{\bf rk}_{{\cal K}_D}(K \setminus Z).
\end{equation}
Since 
$K \setminus U 
=
{\bf cl}_{{\cal K}_D}(K \setminus (Z \setminus J))$, 
we have 
\begin{equation*}
{\bf rk}_{{\cal K}_D}(K \setminus U)
=
{\bf rk}_{{\cal K}_D}(K \setminus (Z\setminus J)).
\end{equation*}
Thus, \eqref{eq_2:critical_set} implies that 
\begin{equation*}
|B| + 
{\bf rk}_{{\cal K}_D}(K \setminus (Z\setminus J)) -
{\bf rk}_{{\cal K}_D}(K \setminus Z)
> {\bf rk}_{{\cal K}_H}(Z). 
\end{equation*}
Furthermore, 
since 
${\bf rk}_{{\cal K}_D}(\cdot)$ is 
a submodular function 
(see, e.g., \cite[Lemma 1.3.1]{Oxley11}),  
we have 
\begin{equation*}
{\bf rk}_{{\cal K}_D}(K \setminus (Z\setminus J)) -
{\bf rk}_{{\cal K}_D}(K \setminus Z) 
\le {\bf rk}_{{\cal K}_D}(Z \cap J). 
\end{equation*}
Thus, since $|Z \cap J| \ge {\bf rk}_{{\cal K}_D}(Z \cap J)$ 
clearly holds, we have  
\begin{equation} \label{eq_2:assumption} 
|B| + |Z \cap J| \ge 
|B| + {\bf rk}_{{\cal K}_D}(Z \cap J) 
> {\bf rk}_{{\cal K}_H}(Z). 
\end{equation}

Define $I$ as an independent set of 
${\cal K}_H$ satisfying the following conditions.
\begin{itemize}
\item
$B \subseteq I \subseteq B \cup (Z \cap J)$. 
\item
For any element $f \in (Z \cap J) \setminus I$, 
$I + f$ is a dependent set of ${\cal K}_H$. 
\end{itemize}
That is, $I$ is a base of ${\cal K}_H|(B \cup (Z \cap J))$
such that $B \subseteq I$. 
Then  
\eqref{eq_2:assumption} implies that 
$(Z \cap J) \setminus I \neq \emptyset$. 
Let $f$ be an element in $(Z \cap J) \setminus I$. 
Define $C \coloneqq \mathbb{C}_{{\cal K}_H}(f,I)$. 

Let $B^{\ast}$ be a base of ${\cal K}_H$ such that 
$I \subseteq B^{\ast}$. 
Notice that $f \notin B^{\ast}$ and 
Lemma~\ref{lemma:closure} implies that 
$f \in {\bf cl}_{{\cal K}_H}(B^{\ast})$. 
Furthermore,  
Lemma~\ref{lemma:elimination} implies that 
$\mathbb{C}_{{\cal K}_H}(f,B^{\ast}) = C$. 
Define $C^{\ast} \coloneqq \mathbb{C}_H(f,B^{\ast})$. 
Then we prove that, for every element $g \in C^{\ast} \setminus J$, 
there exists a circuit $C_g$ of ${\cal M}_H$ such that
$g \in C_g$, $C_g - g \subseteq J$, and $f \notin C_g$. 
If we can prove this, then 
Lemma~\ref{lemma:circuit_union} implies that 
there exists a circuit of ${\cal M}_H$
that is contained in $J$ as a subset.
This contradicts the fact that $J \in {\cal I}_H$. 
In what follows, we fix an element $g \in C^{\ast} \setminus J$. 

First, we consider the case where 
$g \notin {\bf cl}_D(J)$. 
In this case, since $J$ is 
a non-uniformly stable common independent set 
of ${\cal M}_D,{\cal M}_H$, 
$g \in {\bf cl}_H(J)$ and 
$h \succ_H g$ for every element $h \in \mathbb{D}_H(g,J)$. 
Since 
Lemma~\ref{lemma:circuit_matroid_H}
implies that 
$g \succsim_H f$, $f \notin \mathbb{C}_H(g,J)$. 
Thus, we can define $C_g \coloneqq \mathbb{C}_H(g,J)$. 

In the rest of the proof of this case, we assume that 
$g \in {\bf cl}_D(J)$. 
If there exists an element $h \in \mathbb{D}_D(g,J)$
such that $g \succ_D h$, 
then 
since $J$ is 
a non-uniformly stable common independent set 
of ${\cal M}_D,{\cal M}_H$, 
$g \in {\bf cl}_H(J)$ and 
$h \succ_H g$ for every element $h \in \mathbb{D}_H(g,J)$. 
Since 
Lemma~\ref{lemma:circuit_matroid_H}
implies that 
$g \succsim_H f$, $f \notin \mathbb{C}_H(g,J)$. 
Thus, we can define $C_g \coloneqq \mathbb{C}_H(g,J)$. 

In what follows, we assume that 
$h \succsim_D g$ for every element $h \in \mathbb{C}_D(g,J)$. 

\begin{lemma} \label{lemma:not_singleton}
$\{g\} \subsetneq {\bf ta}_D(\mathbb{C}_D(g,J))$. 
\end{lemma} 
\begin{proof} 
Assume that 
$\{g\} = {\bf ta}_D(\mathbb{C}_D(g,J))$.
Then $h \succ_D g$ for 
every element $h \in \mathbb{D}_D(g,J)$. 
Recall that 
$J \subseteq E \setminus P_{z,q-1}$. 
Furthermore, $g \in C^{\ast} \subseteq K \subseteq E \setminus P_{z,q-1}$. 
This and Lemma~\ref{lemma:strong_dominance_choice_D}
imply that $g \notin K$. 
However, this contradicts the fact that $g \in K$. 
\end{proof} 

Lemma~\ref{lemma:not_singleton}
implies that 
$g \in {\bf cl}_H(J)$ and 
$h \succsim_H g$ for every element $h \in \mathbb{C}_H(g,J)$. 
We divide the rest of the proof into the following cases. 
\begin{description}
\item[(1-1)]
$g \notin C$ and $f \in E_1$.
\item[(1-2)]
$g \notin C$ and $f \in E_2$.
\item[(1-3)]
$g \in C$.
\end{description}

{\bf (1-1)} 
In this case, Lemma~\ref{lemma:circuit_matroid_H}
implies that $g \succ_H f$. 
This implies that $f \notin \mathbb{C}_H(g,J)$. 
Thus, we can define $C_g \coloneqq \mathbb{C}_H(g,J)$. 

{\bf (1-2)} 
If $g \in E_2$, then
Lemma~\ref{lemma:circuit_matroid_H}
implies that $g \succ_H f$. 
This implies that $f \notin \mathbb{C}_H(g,J)$, and we can 
define 
$C_g \coloneqq \mathbb{C}_H(g,J)$. 
Assume that 
$g \in E_1$. 
In this case, 
since $J$ is 
a non-uniformly stable common independent set 
of ${\cal M}_D,{\cal M}_H$, 
$g \in {\bf cl}_H(J)$ and 
$h \succ_H g$ for every element $h \in \mathbb{D}_H(g,J)$. 
Thus, $f \notin \mathbb{C}_H(g,J)$, and 
we define $C_g \coloneqq \mathbb{C}_H(g,J)$. 

{\bf (1-3)}
We first prove the following lemma.

\begin{lemma} \label{lemma:g_in_U}
$g \in U$. 
\end{lemma}
\begin{proof}
Since $f \in J$ and $g \notin J$, 
we have $g \neq f$. 
Furthermore, 
$C - f \subseteq I \subseteq B \cup (Z \cap J)$. 
Thus, since $g \in C - f$ and $g \notin J$, 
we have $g \in B \subseteq U$. 
\end{proof} 

In the rest of this proof, 
we prove that $g \in {\bf cl}_{{\cal K}_D}(K \setminus (Z \setminus J))$, 
which contradicts Lemma~\ref{lemma:g_in_U}.
To this end, 
Lemma~\ref{lemma:closure} implies that 
it is sufficient to 
prove that 
there exists a circuit 
$C^{\prime}$ of ${\cal K}_D$ such that 
$g \in C^{\prime}$ and 
$C^{\prime} - g \subseteq K \setminus (Z \setminus J)$. 

Define $C^{\circ} \coloneqq \mathbb{C}_D(g,J)$. 
Then $C^{\circ} \cap (Z \setminus J) = \{g\}$, 
and thus ${\bf ta}_D(C^{\circ}) \cap (Z \setminus J) = \{g\}$. 
Since $J \subseteq E \setminus P_{z,q-1}$, 
Lemma~\ref{lemma:outside_circuit}
implies that 
there exists a circuit $C^{\bullet}$ of ${\cal M}_D$ 
such that 
$C^{\bullet} \subseteq K$ 
and 
$g \in {\bf ta}_D(C^{\bullet}) \subseteq {\bf ta}_D(C^{\circ}) \cap K$.
Thus, since 
\begin{equation*}
g \in 
{\bf ta}_D(C^{\bullet}) \cap (Z \setminus J) \subseteq  
{\bf ta}_D(C^{\circ}) \cap K \cap (Z \setminus J)
= \{g\},
\end{equation*}
we have ${\bf ta}_D(C^{\bullet}) \cap (Z \setminus J) = \{g\}$. 
Lemma~\ref{lemma:tail_circuit_D} implies that 
there exists a circuit $C^{\diamond}$ of 
${\cal K}_D$ such that 
$g \in C^{\diamond} \subseteq {\bf ta}_D(C^{\bullet})$.
Since ${\bf ta}_D(C^{\bullet}) \cap (Z \setminus J) = \{g\}$, 
$C^{\diamond} - g \subseteq K \setminus (Z\setminus J)$. 
This completes the proof. 

\subsection{Case~2}
 
For notational simplicity, we define 
\begin{equation*}
I \coloneqq I_{z,i_z}, \ \ 
b \coloneqq b_z, \ \ 
P \coloneqq P_{z,i_z}, \ \ 
K \coloneqq K_{z,i_z}, \ \ 
{\cal K}_D \coloneqq {\cal M}_D\langle K\rangle, \ \ 
{\cal K}_H \coloneqq {\cal M}_H\langle K\rangle.
\end{equation*} 
Then 
${\bf ta}_H(\mathbb{C}_H(b, I)) - b$ contains a bad element $\xi$. 
The definition of a bad element implies that 
there exists a non-uniformly stable common independent set 
$J$ of ${\cal M}_D,{\cal M}_H$ such that 
$\xi \in J$. 
Since $P$ does not contain a bad element, 
$J \subseteq E \setminus P$. 
Since $b \in P$, $b \notin J$. 

\begin{lemma} \label{lemma:J_b} 
$b \in {\bf cl}_H(J)$ and 
$e \succ_H \xi$ for every element $e \in \mathbb{D}_H(b,J)$. 
\end{lemma}
\begin{proof}
Since $b \in \mathbb{C}_H(b,I)$ and $\xi \in {\bf ta}_H(\mathbb{C}_H(b,I))$,
we have $b \succsim_H \xi$. 

Since $J$ is a non-uniformly stable common independent set 
of ${\cal M}_D,{\cal M}_H$, 
if $b \notin {\bf cl}_D(J)$, then 
$b \in {\bf cl}_H(J)$ and 
$e \succ_H b$ for every element $e \in \mathbb{D}_H(b,J)$.
This implies that since $b \succsim_H \xi$, 
$e \succ_H \xi$ for every element $e \in \mathbb{D}_H(b,J)$.
This completes the proof. 

In what follows, we assume that $b \in {\bf cl}_D(J)$. 
Define $C^{\prime} \coloneqq \mathbb{C}_D(b,J)$. 

\begin{claim} \label{claim_1:lemma:J_b} 
There exists an element $e^{\ast} \in C^{\prime} - b$ such that 
$b \succsim_D e^{\ast}$. 
\end{claim}
\begin{proof}
Assume that $e \succ_D b$ for every element $e \in C^{\prime} - b$. 
Notice that since $J \subseteq E \setminus P$, 
we have $C^{\prime} - b \subseteq E \setminus P$.
Furthermore, since $b \in P$, there exist integers 
$t \in [z]$ and $i \in [i_t]$
such that $b \in K_{t,i}$.
Thus, $b \in E \setminus P_{t,i-1}$. 
Since $P_{t,i-1} \subseteq P$, 
$C^{\prime} \subseteq E \setminus P_{t,i-1}$. 
However, Lemma~\ref{lemma:strong_dominance_choice_D}
implies that $b \notin K_{t,i}$.
This contradicts the fact that $b \in K_{t,i}$. 
This completes the proof. 
\end{proof}

Then Claim~\ref{claim_1:lemma:J_b}  implies that 
$b \in {\bf cl}_H(J)$ and 
$e \succsim_H b$ for every element $e \in \mathbb{C}_H(b,J)$.
If $e \in E_1$, 
then since $J$ is a non-uniformly stable common independent set 
of ${\cal M}_D,{\cal M}_H$,
$e \succ_H b$ for every element $e \in \mathbb{D}_H(b,J)$.
Thus, since $b \succsim_H \xi$, 
$e \succ_H \xi$ for every element $e \in \mathbb{D}_H(b,J)$.

In the rest of this proof, we assume that $b \in E_2$. 
If $b \succ_H \xi$, then 
$e \succ_H \xi$ for every element $e \in \mathbb{D}_H(b,J)$.
Thus, we assume that $b \sim_H \xi$. 
Then $b \in {\bf ta}_H(\mathbb{C}_H(b,I))$. 

If there exists an element $f \in C^{\prime}$ such that 
$b \succ_D f$, then since $J$ is a non-uniformly stable common independent set 
of ${\cal M}_D,{\cal M}_H$,
$e \succ_H b$ for every element $e \in \mathbb{D}_H(b,J)$.
Since $b \succsim_H \xi$, 
$e \succ_H \xi$ for every element $e \in \mathbb{D}_H(b,J)$.
Thus, we assume that $e \succsim_D b$ for every 
element $e \in C^{\prime}$.

\begin{claim} \label{claim_2:lemma:J_b} 
$b \in {\bf cl}_D(I)$ and 
$e \succsim_D b$ for every element $e \in \mathbb{C}_D(b,I)$.
\end{claim}
\begin{proof}
Assume that $C^{\prime} - b \subseteq I$.
Then Lemma~\ref{lemma:closure} implies that 
$b \in {\bf cl}_D(I)$ and 
Lemma~\ref{lemma:elimination} implies that 
$\mathbb{C}_D(b,I) = C^{\prime}$. 
Since $e \succsim_D b$ for every 
element $e \in C^{\prime}$, 
the proof is done. 
In what follows, 
we assume that $C^{\prime} - b \not\subseteq I$. 

Let $f$ be an element in $C^{\prime} \setminus (I+b)$. 
We prove that 
there exists a circuit $C_f$ of ${\cal M}_D$ such that 
$f \in C_f$, 
$C_f - f \subseteq I$, $b \notin C_f$, and 
$e \succsim_D f$ for every element $e \in C_f$. 
If we can prove this, then Lemma~\ref{lemma:circuit_union}
implies that there exists a circuit $C^{\circ}$ 
of ${\cal M}_D$ such that 
$b \in C^{\circ}$, $C^{\circ} - b \subseteq I$, and 
$e \succsim_D b$ for every element $e \in C^{\circ}$. 
Since Lemma~\ref{lemma:elimination} implies that 
$\mathbb{C}_D(b,I) = C^{\circ}$, this completes the 
proof. 

If $f \in K$, then 
Lemma~\ref{lemma:main_alg:base_D} implies that 
$f \in {\bf cl}_{{\cal K}_D}(I)$. 
Define $C_f \coloneqq \mathbb{C}_D(f,I)$. 
Then $C_f - f \subseteq I$. 
Since $f \neq b$, $C_f - f \subseteq I$, and $b \notin I$, 
we have $b \notin C_f$. 
Lemma~\ref{lemma:circuit_matroid_D}
implies that $e \succsim_D f$ for every element $e \in C_f$.
This completes the proof of this case. 

If $f \notin K$, then 
Lemma~\ref{lemma:outside_choice_D} implies that 
there exists a circuit $C^{\bullet}$ of ${\cal M}_D$ such that 
$f \in C^{\bullet}$, $C^{\bullet} - f \subseteq K$, and 
$e \succ_D f$ for every element $e \in C^{\bullet} - f$. 
Since $b \in P$, $b \notin K$.
This implies that $b \notin C^{\bullet}$.
Let $g$ be an element in $C^{\bullet} \setminus (I+f)$. 
Since $g \in K$, 
Lemma~\ref{lemma:main_alg:base_D} implies that 
$g \in {\bf cl}_{{\cal K}_D}(I)$. 
Define $C^{\diamond} \coloneqq \mathbb{C}_D(g,I)$. 
Then $C^{\diamond} - g \subseteq I$. 
Since $g \neq b$, $C^{\diamond} - g \subseteq I$, and $b \notin I$, 
we have $b \notin C^{\diamond}$. 
Furthermore, Lemma~\ref{lemma:circuit_matroid_D}
implies that $e \succsim_D g$ for every element $e \in C^{\diamond}$.
Since $g \succ_D f$, $f \notin C^{\diamond}$. 
Thus, Lemma~\ref{lemma:circuit_union}
implies that 
there exists a circuit $C_f$ of ${\cal M}_D$ such that 
$f \in C_f$, 
$C_f - f \subseteq I$, 
$b \notin C_f$, and 
$e \succsim_D f$ for every element $e \in C_f$. 
This completes the proof. 
\end{proof} 
Recall that $b \in {\bf block}(I)$.
Furthermore, $b \in E_2$ and 
$b \in {\bf ta}_H(\mathbb{C}_H(b,I))$.
Thus, 
$b \notin {\bf cl}_D(I)$, or 
$b \in {\bf cl}_D(I)$ and 
there exists 
an element $g \in \mathbb{D}_D(b,I)$ 
such that $b \succ_D g$. 
However, this contradicts Claim~\ref{claim_2:lemma:J_b}.
This completes the proof. 
\end{proof} 

We consider a maximum-size common independent set $L$ 
of ${\cal K}_D, {\cal K}_H$ satisfying 
the following conditions. 
\begin{description}
\item[(C1)]
$b \in {\bf cl}_H(L) \setminus L$. 
\item[(C2)]
$\xi \succsim_H e$ for every element $e \in {\bf ta}_H(\mathbb{C}_H(b,L))$. 
\item[(C3)]
$J \cap {\bf ta}_H(\mathbb{C}_H(b,L)) \neq \emptyset$. 
\item[(C4)]
$L$ maximizes $|J \cap L|$ among all the  
maximum-size common independent sets  
of ${\cal K}_D, {\cal K}_H$ satisfying 
(C1), (C2), and (C3). 
\end{description} 
Since $I$ satisfies (C1), (C2), and (C3), 
$L$ is well-defined. 

\begin{lemma} \label{lemma:case_2_strict} 
Let $e$ be an element in $\mathbb{D}_H(b,L) \setminus J$
such that $e \in {\bf cl}_D(J)$. 
Then 
there exists an element $f \in \mathbb{D}_D(e,J)$ 
such that $e \succ_D f$. 
\end{lemma}

Before proving 
Lemma~\ref{lemma:case_2_strict}, 
we first prove that 
Lemma~\ref{lemma:case_2_strict}
completes the proof of this case. 
Let $e$ be an element in $\mathbb{C}_H(b,L) \setminus (J+b)$, and 
let $\eta$ be an element in $J \cap {\bf ta}_H(\mathbb{C}_H(b,L))$. 
Then since $J$ is 
a non-uniformly stable common independent set 
of ${\cal M}_D,{\cal M}_H$, 
Lemma~\ref{lemma:case_2_strict}  implies that 
$e \in {\bf cl}_H(J)$ and 
$f \succ_H e$ for every element $f \in \mathbb{D}_H(e,J)$. 
Since 
$\eta \in {\bf ta}_H(\mathbb{C}_H(b,L))$, 
$e \succsim_H \eta$. 
This implies that $\eta \notin \mathbb{C}_H(e,J)$.
Since $\xi \succsim_H \eta$, 
Lemma~\ref{lemma:J_b} implies that 
$\eta \notin \mathbb{C}_H(b,J)$. 
Thus, 
Lemma~\ref{lemma:circuit_union} implies 
that there exists a circuit $C^{\prime}$ of ${\cal M}_H$ 
such that 
$C^{\prime} \subseteq J$. 
This completes the proof. 

What remains is to prove Lemma~\ref{lemma:case_2_strict}.
To this end, we first prove a necessary lemma. 

Let $e$ be an element in $L \setminus J$
such that $e \in {\bf cl}_D(J)$. 
Then since $L \subseteq K$, we have $e \in K$. 
This implies that, 
since $J \subseteq E \setminus P$, 
Lemma~\ref{lemma:strong_dominance_choice_D}
implies that 
there exists an element $f \in \mathbb{D}_D(e,J)$ 
such that $e \succsim_D f$. 
Thus, 
since $J$ is 
a non-uniformly stable common independent set 
of ${\cal M}_D,{\cal M}_H$, 
$e \in {\bf cl}_H(J)$ and 
$f \succsim_H e$ for every element $f \in \mathbb{C}_H(e,J)$. 

\begin{lemma} \label{lemma:case_2_auxiliary} 
Let $e$ be an element in $L \setminus J$
such that $e \in {\bf cl}_D(J)$. 
Assume that 
$g \succsim_D e$ 
for every element $g \in \mathbb{C}_D(e,J)$.
Then the following statements hold. 
\begin{description}
\item[(D1)]
There exists an element $f \in J \setminus L$ such that 
$f \in {\bf cl}_{{\cal K}_D}(L)$ and 
$e \in \mathbb{C}_{{\cal K}_D}(f,L)$.
\item[(D2)]
$f \in {\bf cl}_{{\cal K}_H}(L)$, and 
there exists an element $g \in L \setminus J$ such that 
$g \in \mathbb{C}_{{\cal K}_H}(f,L)$, 
$g \in {\bf cl}_D(J)$, and
$h \succsim_D g$ for every 
element $h \in \mathbb{C}_D(g,J)$.  
\end{description}
\end{lemma} 
\begin{proof}
{\bf (D1)} 
For notational simplicity, we define $\Delta \coloneqq J \setminus L$. 

We prove that 
there exists 
a circuit $C$ of ${\cal K}_D$ such that 
$e\in C \subseteq L \cup J$ and $|C \cap \Delta|=1$.
Let $f$ be the element in $C \cap \Delta$. 
Then 
Lemma~\ref{lemma:closure} implies that 
$f \in {\bf cl}_{{\cal K}_D}(L)$.
Furthermore, 
Lemma~\ref{lemma:elimination} implies that 
$C = \mathbb{C}_{{\cal K}_D}(f,L)$. 
This completes the proof of (D1). 

Define $C_1 \coloneqq \mathbb{C}_D(e,J)$. 
Then Lemma~\ref{lemma:outside_circuit}
implies that 
there exists a circuit $C_2$ of ${\cal M}_D$ 
such that 
$C_2 \subseteq K$ and 
$e \in {\bf ta}_D(C_2) \subseteq {\bf ta}_D(C_1) \cap K \subseteq L \cup J$.
Thus, since $e \in {\bf ta}_D(C_2)$, 
Lemma~\ref{lemma:tail_circuit_D}
implies that 
there exists a circuit $C_3$ of ${\cal K}_D$ 
such that $e \in C_3 \subseteq {\bf ta}_D(C_2)$. 
Let $C_4$ be a circuit 
of ${\cal K}_D$ 
such that $e \in C_4 \subseteq L \cup J$ 
and it minimizes $|C_4 \cap \Delta|$ among 
all the circuits $C^{\prime}$ of ${\cal K}_D$ 
such that $e \in C^{\prime} \subseteq L \cup J$. 
The existence of $C_3$ implies that 
$C_4$ is well-defined. 
If $C_4 \subseteq L$, then 
this contradicts the fact that $L$ is an independent set of ${\cal K}_D$. 
Thus, $C_4 \cap \Delta \neq \emptyset$. 
If $|C_4 \cap \Delta| = 1$, 
then 
the proof is done.
Thus, we assume that  
$|C_4 \cap \Delta| \ge 2$. 
Since Lemma~\ref{lemma:main_alg:base_D} implies that 
$I$ is a base of ${\cal K}_D$ and $|I|=|L|$, 
$L$ is a base of ${\cal K}_D$. 
Thus, 
$g \in {\bf cl}_{{\cal K}_D}(L)$ 
for every element $g \in C_4 \cap \Delta$. 
If there exists an element $g \in C_4 \cap \Delta$ such that 
$e \notin \mathbb{C}_{{\cal K}_D}(g,L)$, 
Lemma~\ref{lemma:elimination} implies that 
there exists a circuit $C^{\prime}$  
of ${\cal K}_D$ 
such that $e \in C^{\prime} \subseteq L \cup J$ and 
$C^{\prime} \cap \Delta \subseteq 
(C_4 \cap \Delta) - g$.
This contradicts the definition of $C_4$. 
Thus, $e \in \mathbb{C}_{{\cal K}_D}(g,L)$ 
for 
every element $g \in C_4 \cap \Delta$. 
Let $h,h^{\prime}$ be elements in 
$C_4 \cap \Delta$
such that $h \neq h^{\prime}$. 
Then
Lemma~\ref{lemma:elimination}
implies that there exists 
a circuit $C_5$ of ${\cal K}_D$
such that 
$h \in C_5 \subseteq L \cup J$,
$e \notin C_5$, and 
$C_5 \cap \Delta \subseteq \{h,h^{\prime}\} \subseteq C_4 \cap \Delta$. 
This and 
Lemma~\ref{lemma:elimination} imply
that 
there exists a circuit $C^{\prime}$ of 
${\cal K}_D$ such that $e \in C^{\prime} \subseteq L \cup J$
and 
$C^{\prime} \cap \Delta \subseteq 
(C_4 \cap \Delta) - h$.
This contradicts the definition of $C_4$. 
This completes the proof of (D1). 

{\bf (D2)} 
Since 
Lemma~\ref{lemma:main_alg:base_H} 
implies that $L$ is a base of ${\cal K}_H$, 
$f \in {\bf cl}_{{\cal K}_H}(L)$. 

Assume that, for every element $g \in \mathbb{C}_{{\cal K}_H}(f,L)$, 
(i) $g \notin {\bf cl}_D(J)$, or 
(ii) $g \in {\bf cl}_D(J)$ and 
there exists an element $h \in \mathbb{D}_D(g,J)$ 
such that $g \succ_D h$. 

Define $C \coloneqq \mathbb{C}_H(f,L)$. 
We prove that, for every element $g \in C \setminus J$, 
there exists a circuit $C_g$ of ${\cal M}_H$ such that 
$g \in C_g$, 
$C_g - g \subseteq J$, and 
$f \notin C_g$.
This and 
Lemma~\ref{lemma:circuit_union}
imply that 
there exists a circuit $C^{\prime}$ of ${\cal M}_H$ 
such that $C^{\prime} \subseteq J$.
However, this contradicts the fact that $J \in {\cal I}_H$. 
Let $g$ be an element in $C \setminus J$. 

First, we assume that $g \notin {\bf cl}_D(J)$.
Since $J$ is 
a non-uniformly stable common independent set 
of ${\cal M}_D,{\cal M}_H$, 
$g \in {\bf cl}_H(J)$ and 
$h \succ_H g$ for every element $h \in \mathbb{D}_H(g,J)$. 
Since 
$g \sim_H f$ follows from Lemma~\ref{lemma:circuit_matroid_H}, 
we can define $C_g \coloneqq \mathbb{C}_H(g,J)$. 
From here, we assume that 
$g \in {\bf cl}_D(J)$.

Next, we assume that 
$g \in \mathbb{C}_{{\cal K}_H}(f,L)$.
Then 
the above assumption 
implies that 
there exists an element $h \in \mathbb{D}_D(g,J)$ 
such that 
$g \succ_D h$. 
This implies that $g \in {\bf cl}_H(J)$ and 
$h \succ_H g$ for every element $h \in \mathbb{D}_H(g,J)$. 
Since $g \sim_H f$, we can define $C_g \coloneqq \mathbb{C}_H(g,J)$. 

Third, we assume that 
$g \in C \setminus \mathbb{C}_{{\cal K}_H}(f,L)$. 
If there exists an element $h \in \mathbb{D}_D(g,J)$ 
such 
that $g \succ_D h$, then 
$g \in {\bf cl}_H(J)$ and 
$h \succ_H g$ for every element $h \in \mathbb{D}_H(g,J)$. 
This implies that we can define $C_g \coloneqq \mathbb{C}_H(g,J)$. 
In what follows, we assume that 
$h \succsim_D g$ for every element $h \in \mathbb{C}_D(g,J)$. 
Since $g \in K$, 
Lemma~\ref{lemma:strong_dominance_choice_D}
implies that ${\bf ta}_D(\mathbb{C}_D(g,J)) \neq \{g\}$. 
This implies that if $g \in E_1$, then since $J$ is 
a non-uniformly stable common independent set 
of ${\cal M}_D,{\cal M}_H$, 
$g \in {\bf cl}_H(J)$ and 
$h \succ_H g$ for every element $h \in \mathbb{D}_H(g,J)$.
Thus, in this case, 
we can define $C_g \coloneqq \mathbb{C}_H(g,J)$. 
Assume that $g \in E_2$. 
Since $J$ is 
a non-uniformly stable common independent set 
of ${\cal M}_D,{\cal M}_H$, 
$g \in {\bf cl}_H(J)$ and 
$h \succsim_H g$ for every element $h \in \mathbb{C}_H(g,J)$.
If $f \in E_2$, then 
since $g,f \in E_2$, 
Lemma~\ref{lemma:circuit_matroid_H}
implies 
$g \succ_H f$. 
If $f \in E_1$, then  
Lemma~\ref{lemma:circuit_matroid_H}
implies 
$g \succ_H f$. 
Thus,
we can define $C_g \coloneqq \mathbb{C}_H(g,J)$. 
This completes the proof. 
\end{proof} 

We are now ready to prove 
Lemma~\ref{lemma:case_2_strict}. 

\begin{proof}[Proof of Lemma~\ref{lemma:case_2_strict}]
Assume that 
there exists an element $e \in \mathbb{D}_H(b,L) \setminus J$
such that 
$e \in {\bf cl}_D(J)$
and 
$f \succsim_D e$
for every element $f \in \mathbb{C}_D(e,J)$. 
Since $E$ is finite, 
Lemma~\ref{lemma:case_2_auxiliary} implies that 
there exists a directed cycle 
in ${\bf G}_{{\cal K}_D{\cal K}_H}(L)$.  
Thus, 
there exists a simple directed cycle 
in ${\bf G}_{{\cal K}_D{\cal K}_H}(L)$
for which there does not a shortcut arc. 
Assume that this simple directed cycle passes through 
elements 
$e_1,f_1,e_2,f_2,\dots,e_{\ell},f_{\ell} \in K$ in this order and 
satisfies the following conditions. 
\begin{itemize}
\item
$e_1,e_2,\dots,e_{\ell} \in J \setminus L$ and 
$f_1,f_2,\dots,f_{\ell} \in L \setminus J$. 
\item
$f_i \in \mathbb{C}_{{\cal K}_H}(e_i,L)$
for every integer $i \in [\ell]$. 
\end{itemize}
For each integer $i \in [\ell]$, 
we define 
\begin{equation*}
M_i \coloneqq (L \setminus \{f_1,f_2,\dots,f_i\}) \cup \{e_1,e_2,\dots,e_i\}. 
\end{equation*} 
Define $M_0 \coloneqq L$. 
Then we prove that, for every 
integer $i \in [\ell] \cup \{0\}$, 
$M_i$ satisfies (C1), (C2), and (C3). 
Lemma~\ref{lemma:IriT76_intersection}
implies that 
$M_{\ell}$ is 
a maximum-size common independent set 
of ${\cal K}_D, {\cal K}_H$. 
Since $|J \cap M_{\ell}| > |J \cap L|$, 
this contradicts the definition of $L$. 
This completes the proof. 

Since 
$L$ is an independent set of ${\cal K}_H$, 
Lemma~\ref{lemma:IriT76_sequence}
implies that 
$M_i$ is an independent set of ${\cal K}_H$
for every integer $i \in [\ell]$. 
Thus, $f_i \in \mathbb{C}_{{\cal K}_H}(e_i,M_{i-1})$
for every integer $i \in [\ell]$. 

Let $i$ be an integer in $[\ell]$. 
Assume that 
$M_{i-1}$ satisfies (C1), (C2), and (C3).
Then we prove that 
$M_i$ satisfies (C1), (C2), and (C3). 
We can prove that $M_i$ satisfies (C1) as follows. 
Since $M_{i-1}$ satisfies (C1), 
Lemma~\ref{lemma:IriT76_sequence}
implies that it is sufficient to prove that 
$b \neq e_i$. 
This follows from the fact that $b \notin J$. 

In what follows, we prove that 
$M_i$ satisfies (C2) and (C3). 
Define $C_{i-1} \coloneqq \mathbb{C}_H(b,M_{i-1})$ and 
$C_i \coloneqq \mathbb{C}_H(b,M_i)$. 
Let $\eta$ be an element in 
$J \cap {\bf ta}_H(C_{i-1})$.
We divide the proof into the following three cases. 
\begin{description}
\item[(2-1)]
$f_i \succ_H \eta$.
\item[(2-2)]
$\eta \succ_H f_i$.
\item[(2-3)]
$f_i \sim_H \eta$.
\end{description}
Notice that since $f_i \in \mathbb{C}_{{\cal K}_H}(e_i,L)$, 
Lemma~\ref{lemma:circuit_matroid_H} implies 
that $e_i \sim_H f_i$. 

{\bf (2-1)} 
For each integer $j \in \{i,i-1\}$, 
we define $C^{\prime}_j$ as the set of 
elements $g \in C_j$ such that 
$\eta \succsim_H g$. 
Since $\eta \in {\bf ta}_H(C_{i-1})$, 
$C_{i-1}^{\prime} = {\bf ta}_H(C_{i-1})$.
Thus, for every pair of elements $g,g^{\prime} \in C^{\prime}_{i-1}$, 
we have $g \sim_H g^{\prime}$.
Claim~\ref{claim_1:lemma:case_2_strict}
implies that 
if we can prove that $C^{\prime}_{i-1} = C^{\prime}_i$, then 
since 
$M_{i-1}$ satisfies (C2) and (C3), 
$M_i$ satisfies (C2) and (C3). 

\begin{claim} \label{claim_1:lemma:case_2_strict}
Assume that $C^{\prime}_{i-1} = C^{\prime}_i$.
Then  
${\bf ta}_H(C_{i-1}) = {\bf ta}_H(C_i)$. 
\end{claim}
\begin{proof}
Let $g$ be an element in ${\bf ta}_H(C_{i-1}) = C^{\prime}_{i-1}$.
Then since 
$C^{\prime}_{i-1} = C^{\prime}_i$, $g \in C^{\prime}_i$.  
Assume that $g \notin {\bf ta}_H(C_i)$. 
Then there exists an element 
$h \in C_i$ such that 
$g \succ_H h$.
Since $g \in C^{\prime}_i$, $\eta \succsim_H g$.  
Thus, $\eta \succsim_H h$.
This implies that $h \in C^{\prime}_i$. 
Since $C^{\prime}_{i-1} = C^{\prime}_i$, 
$h \in C^{\prime}_{i-1}$. 
However, this contradicts the fact that $g \sim_H h$. 

Let $g$ be an element in ${\bf ta}_H(C_i)$. 
Then $h \succsim_H g$ 
for every element $h \in C_i$.
Thus,
if $g \notin C^{\prime}_i$, then 
$C^{\prime}_i = \emptyset$. 
However, since $\eta \in C^{\prime}_{i-1} = C^{\prime}_i$, 
$C^{\prime}_i \neq \emptyset$. 
This is a contradiction. 
This implies that 
$g \in C^{\prime}_i = C^{\prime}_{i-1} = {\bf ta}_H(C_{i-1})$. 
This completes the proof. 
\end{proof}

Since $e_i \sim_H f_i \succ_H \eta$, 
we have $f_i \notin C^{\prime}_{i-1}$ and 
$e_i \notin C^{\prime}_i$. 
Thus, in order to prove that 
$C^{\prime}_{i-1} = C^{\prime}_i$, 
it suffices to prove that,
for every element $g \in M_{i-1}$ such that 
$\eta \succsim_H g$, 
$g \notin \mathbb{C}_H(e_i,M_{i-1})$.
If we can prove this, then 
Lemma~\ref{lemma:IriT76_sequence}
implies that, 
for every element $g \in M_{i-1} - f_i$ such that 
$\eta \succsim_H g$, 
$g \in C^{\prime}_{i-1}$ if and only if 
$g \in C^{\prime}_i$. 

Since $e_i \in {\bf cl}_{{\cal K}_H}(M_{i-1})$, 
Lemma~\ref{lemma:circuit_matroid_H} implies that, for every 
element $h \in \mathbb{C}_H(e_i,M_{i-1})$, 
$h \succsim_H e_i$. 
Thus, for every element $g \in M_{i-1}$ such that 
$\eta \succsim_H g$,
since $e_i \succ_H \eta \succsim_H g$, 
$g \notin \mathbb{C}_H(e_i,M_{i-1})$. 

{\bf (2-2)}
In this case, we prove that $C_{i-1} = C_i$. 
If we can prove this, then since $M_{i-1}$ satisfies (C2) and (C3), 
$M_i$ satisfies (C2) and (C3). 
Notice that since $\eta \in {\bf ta}_H(C_{i-1})$ and 
$\eta \succ_H f_i$, we have $f_i \notin C_{i-1}$. 
By using this, 
we can prove that $C_{i-1} = C_i$
as follows. 
Lemma~\ref{lemma:IriT76_auxiliary}
implies that 
$e_i \notin C_i$. 
Furthermore, 
Lemma~\ref{lemma:IriT76_auxiliary}
implies that 
for every element $g \in M_{i-1} - f_i$, 
$g \in C_{i-1}$ if and only if $g \in C_i$.
This completes the proof. 

{\bf (2-3)} 
If $f_i \notin C_{i-1}$, then 
$f_i \neq \eta$ and 
Lemma~\ref{lemma:IriT76_auxiliary}
implies that $\eta \in C_i$. 
In addition, if $f_i \in C_{i-1}$, then 
Lemma~\ref{lemma:IriT76_auxiliary}
implies that $e_i \in C_i$. 
Recall that $e_i \sim_H \eta$. 
Thus, since $\xi \succsim_H \eta \sim_H e_i$ and 
$e_i \in J$, 
it suffices to prove that, for every 
element $g \in M_{i-1}$ such that 
$\eta \succ_H g$, we have 
$g \notin \mathbb{C}_H(e_i,M_{i-1})$. 
Lemma~\ref{lemma:circuit_matroid_H} implies that, for every 
element $h \in \mathbb{C}_H(e_i,M_{i-1})$, 
$h \succsim_H e_i$. 
Thus, 
for every 
element $g \in M_{i-1}$ such that 
$\eta \succ_H g$, we have 
$g \notin \mathbb{C}_H(e_i,M_{i-1})$. 
\end{proof} 

\bibliographystyle{plain}% bib style
\bibliography{nonuniform_common_bib}

\end{document}